\newif\ifrenderfigures
\renderfigurestrue

\documentclass[11pt,english]{article}
\usepackage[margin=1in]{geometry}

\usepackage{mdwlist}
\usepackage{enumerate}
\usepackage{amssymb} %,amsbsy,latexsym}
\usepackage{amsmath,amscd,amsthm} % AMS Math packages
\usepackage{graphics, subfigure, float}
\usepackage{fp, calc}
\usepackage{hyperref}
\usepackage{url}

\usepackage[T1]{fontenc} % Wichtig: Muss vor dem Laden der Fonts eingebunden werden
\usepackage{bm}

\usepackage{pst-all}
\usepackage{pstricks-add}
\usepackage{pst-func}
\usepackage{verbatim, comment}
\usepackage{datetime}

\theoremstyle{theorem}
\newtheorem{theorem}{Theorem}

\newtheorem{corollary}[theorem]{Corollary}
\newtheorem{lemma}[theorem]{Lemma}

\newtheorem{remark}[theorem]{Remark}
 \newtheorem{definition}[theorem]{Definition}

\def\rem#1{{\marginpar{\raggedright\scriptsize #1}}}

\providecommand{\setN}{\mathbb{N}}
\providecommand{\setZ}{\mathbb{Z}}

\providecommand{\setR}{\mathbb{R}}

\newcommand{\E}{\mathop{\mathbb{E}}}
\newcommand{\p}{\mathsf{P}}
\newcommand{\Prec}{\mathsf{prec}}
\newcommand{\Intervals}{c\mathsf{-intervals}}
\newcommand{\dup}{\mathsf{dup}}
\newcommand{\Cmax}{C_{\mathsf{max}}}
\newcommand{\wjcj}{\sum w_j C_j}

        \def\drawRect#1#2#3#4#5{
           \FPeval{\x2}{(#2) + #4} 
           \FPeval{\y2}{(#3) + #5} 
           \pspolygon[#1](#2,#3)(\x2,#3)(\x2,\y2)(#2,\y2)
        }

\usepackage{calrsfs} % now e.g. \pazocal{I} gives a nice I
\DeclareMathAlphabet{\pazocal}{OMS}{zplm}{m}{n}
\usepackage[displaymath,textmath,graphics, subfigure, floats]{preview} %sections,
\PreviewEnvironment{center} 

\makeatother

\title{Scheduling with Communication Delays \\ via LP Hierarchies and Clustering}
\author{Sami Davies\thanks{University of Washington, Seattle. Email: {\tt\{daviess,rothvoss,yihaoz93\}@uw.edu}. Thomas Rothvoss is supported by NSF CAREER grant 1651861 and a David \& Lucile Packard Foundation Fellowship.} \qquad Janardhan Kulkarni\thanks{Microsoft Research, Redmond. Email: {\tt\{jakul,jatarnaw\}@microsoft.com}.} \qquad Thomas Rothvoss\footnotemark[1] \\ Jakub Tarnawski\footnotemark[2] \qquad Yihao Zhang\footnotemark[1]}

\date{\today, \currenttime}

\begin{document}

\maketitle

\begin{abstract}
  We consider the classic problem of scheduling jobs with precedence constraints on identical machines
  to minimize  makespan,
  in the presence of {\em communication delays}.
  In this setting,
  denoted by $\p \mid \Prec, c \mid \Cmax$,
  if two dependent jobs are scheduled on different machines,
  then at least $c$ units of time must pass between their executions.
  Despite its relevance to many applications, this model remains one of the most poorly understood  in scheduling theory.
  Even for a special case where an unlimited number of machines is available, the best known approximation ratio is
  $2/3 \cdot (c+1)$, whereas  Graham's greedy list scheduling algorithm already gives a $(c+1)$-approximation in that setting.
  An outstanding open problem in the top-10 list by Schuurman and Woeginger and its recent update by Bansal
  asks whether there exists a constant-factor approximation algorithm.
  
  In this work we give a polynomial-time $O(\log c \cdot \log m)$-approximation algorithm
  for this problem,
  where $m$ is the number of machines and $c$ is the communication delay.
  Our approach is based on a Sherali-Adams lift of a linear programming relaxation
  and a randomized clustering of the semimetric space induced by this lift.

\end{abstract}

\section{Introduction}

Scheduling jobs with precedence constraints is a fundamental problem in approximation algorithms and combinatorial optimization.
In this problem we are given $m$ identical  machines and a set  $J$ of $n$ jobs, where each job $j$ has a processing length $p_j  \in \setZ_+$.
The jobs have precedence constraints, which are given by a partial order $\prec$. 
A constraint $j \prec j'$ encodes that job $j'$ can only start after job $j$ is completed.
The goal is to find a schedule of jobs that minimizes {\em makespan}, which is the completion time of the last job.
This problem is denoted\footnote{
	Throughout the paper we use the standard scheduling three-field notation~\cite{GLLR79,VeltmanLL90}.
	The respective fields
	denote:
	\textbf{(1)
	number of identical machines:} $\p \infty$: unlimited; $\p$: number $m$ of machines given as input; $\p m$: constant number $m$ of machines,
	\textbf{(2) job properties:} $\Prec$: precedence constraints; $p_j=1$: unit-size jobs; $c$:~communication delays of length $c$ (can be $c_{jk}$ if dependent on jobs $j \prec k$); $\Intervals$: see Section~\ref{sec:approximation_for_pinfty}; $\dup$: allowed duplication of jobs,
	\textbf{(3) objective:} $\Cmax$: minimize makespan; $\wjcj$: minimize weighted sum of completion times.
%	\end{itemize}%
}
by $\p \mid \Prec \mid \Cmax$.
In a seminal result from 1966, Graham ~\cite{GrahamListScheduling1966} showed that the greedy list scheduling algorithm achieves a $\left(2- \frac{1}{m}\right)$-approximation. 
%The algorithm simply  computes an arbitrary topological ordering of the jobs, and whenever a machine becomes idle selects the first
%available job from the list. 
By now, our understanding of the approximability of this basic problem is almost complete: 
it had been known since the late `70s, due to a result by Lenstra and Rinnooy Kan~\cite{LR78}, that it is NP-hard to obtain better than $4/3$-approximation, and in 2010 Svensson~\cite{Svensson10} showed that, assuming a variant of the Unique Games Conjecture~\cite{BansalK10}, it is NP-hard to get a $(2-\varepsilon)$-approximation for any $\varepsilon > 0$.

The above precedence-constrained scheduling problem models the task of distributing workloads onto multiple processors or servers, which is ubiquitous in computing.
This basic setting takes the dependencies between work units into account, but not the data transfer costs between machines, which is critical in applications. 
A precedence constraint $j \prec j'$ typically implies that the input to $j'$ depends on the output of $j$.
In many real-world scenarios, especially in the context of scheduling in data centers, if $j$ and $j'$ are executed on different machines, then the {\em communication delay} due to transferring this output to the other machine cannot be ignored.
This is an active area of research in applied data center scheduling literature, where several new abstractions have been proposed to deal with communication delays~\cite{Chowdhury, guo2012spotting,hong2012finishing,shymyrbay2018meeting,zhang2012optimizing,zhao2015rapier,luo2016towards}.
 Another timely example is found in the parallelization of Deep Neural Network training
(the machines being accelerator devices such as GPUs, TPUs, or FPGAs).
There, when training the network on one sample/minibatch per device in parallel, the communication costs incurred by synchronizing the weight updates in fact dominate the overall running time~\cite{narayanan2018pipedream}.
Taking these costs into account,
it turns out that
it is better to split the network onto multiple devices, forming a ``model-parallel'' computation pipeline~\cite{huang2019gpipe}.
In the resulting \emph{device placement} problem,
the optimal split crucially depends on the communication costs
between dependent layers/operators.

\medskip

A classic model that captures the effect of data transfer latency on scheduling decisions is the problem of {\em scheduling jobs with precedence and communication delay constraints}, introduced by Rayward-Smith~\cite{RAYWARDSMITH1987} and Papadimitriou and Yannakakis~\cite{PapadimitriouY90}.
The setting, denoted by $\p \mid \Prec, c \mid \Cmax$, is similar to the makespan minimization problem described earlier, except for one crucial difference.
Here  we are  given a {\em communication delay parameter} $c \in  \setZ_{\ge 0}$, and the output schedule must satisfy the property that if $j \prec j'$ and $j$, $j'$ are scheduled on {different machines}, then $j'$ can only start executing at least $c$ time units after $j$ had finished. 
On the other hand,  if $j$ and $j'$ are scheduled on the {same machine}, then $j'$ can start executing immediately after $j$ finishes. 
In a closely related problem, denoted by $\p \infty \mid \Prec, c \mid \Cmax$, a schedule can use as many machines as desired.
The goal is to schedule jobs {\em non-preemptively} so as to minimize the makespan.  In a non-preemptive schedule, each job $j$ needs to be  assigned to a single machine and executed during $p_j$ consecutive timeslots.
The problems $\p \mid \Prec, c \mid \Cmax$ and $\p \infty \mid \Prec, c \mid \Cmax$ are the focus of this paper.

Despite its theoretical significance and practical relevance, very little is known about the communication delay setting.
A direct application of Graham's~\cite{GrahamListScheduling1966} list scheduling algorithm yields a $(c+2)$-approximation,
and no better algorithm is known for the problem.
Over the years, the problem has attracted significant attention, but all known results,
which we discuss below in Section~\ref{sec:related_work}, concern special settings, small communication delays, or hardness of approximation. 
To put this in perspective, we note that the current best algorithm for general~$c$~\cite{GiroudeauKMP08}, which achieves an approximation factor of $2/3 \cdot (c+1)$, only marginally improves on Graham's algorithm while requiring the additional assumptions that the number of machines is unbounded and $p_j = 1$.
This is in sharp contrast to the basic problem  $\p \mid \Prec \mid \Cmax$ (which would correspond to the case $c=0$), where the approximability of the problem is completely settled under a variant of the Unique Games Conjecture.
This situation hints that incorporating communication delays in scheduling  decisions requires fundamentally new algorithmic ideas compared to the no-delay setting.
Schuurman and Woeginger~\cite{SW99a} placed the quest for getting better algorithms to the problem in their influential list of top-10 open problems in scheduling theory.
In a recent MAPSP 2017 survey talk,  Bansal~\cite{Bansalmapsp} highlighted the lack of progress on this model, describing it as ``not understood at all; almost completely open'', and suggested that this is due to the lack of promising LP/SDP relaxations.

\subsection{Our Contributions}

The main result of this paper is the following:

\begin{theorem} \label{thm:main}
There is a randomized $O(\log c \cdot \log m)$-approximation algorithm for $\p \mid \Prec, c \mid \Cmax$ with expected polynomial running time, where $c, p_j \in \setN$.
\end{theorem}

In any non-preemptive schedule the number $m$ of machines is at most the number $n$ of jobs, so for the easier $\p \infty$ version of the problem, the above theorem implies the following:

\begin{corollary} \label{cor:colmain}
There is a randomized $O(\log c \cdot \log n)$-approximation algorithm for $\p \infty \mid \Prec, c \mid \Cmax$ with expected polynomial running time, where $c, p_j \in \setN$.
\end{corollary}

For both problems one can replace either $c$ or $m$ by $n$, yielding a $O(\log^2 n)$-approximation algorithm.
Our results make substantial progress towards resolving  one of the questions in ``Open Problem 3'' in the survey of Schuurman and Woeginger~\cite{SW99a}, which asks whether a constant-factor approximation algorithm exists for $\p \infty \mid \Prec, c \mid \Cmax$.

Our approach is based on a Sherali-Adams lift of a time-indexed linear programming relaxation for the problem, followed by a randomized clustering of the semimetric space induced by this lift.
%To our knowledge, it is the first application of clustering in metric spaces to scheduling problems.
To our knowledge, this is the first instance of a multiple-machine scheduling problem being viewed via the lens of metric space clustering.
We believe that our framework is fairly general and should extend to other problems involving scheduling with communication delays.
To demonstrate the broader applicability of our approach, we also consider the objective of minimizing the weighted sum of completion times. Here each job $j$ has a weight $w_j$, and the goal is to minimize $\sum_{j} w_j C_j$, where $C_j$ is the completion time of $j$.  

\begin{theorem} \label{thm:maincomp}
There is a randomized $O(\log c \cdot \log n)$-approximation algorithm for  $\p \infty \mid \Prec, p_j = 1, c \mid \sum_{j} w_j C_j$ with expected polynomial running time, where $c \in \setN$.
\end{theorem}

No non-trivial approximation ratio was known for this problem prior to our work.

\subsection{Our Techniques}

% \begin{center} {\bf TODO} \end{center}

As we alluded earlier, there is a lack of
combinatorial lower bounds for scheduling with communication delays. For example, consider {Graham's list scheduling algorithm}, which greedily processes jobs on $m$ machines as soon as they become available. One can revisit the analysis of Graham~\cite{GrahamListScheduling1966} and show that there exists a chain $Q$ of dependent jobs such that the makespan achieved by list scheduling is bounded by

\[
  \frac{1}{m} \sum_{j \in J} p_j + \sum_{j \in Q} p_j + c \cdot (|Q|-1). 
\]
The first two terms are each lower bounds on the optimum --- the 3rd term is not. In particular, it is unclear 
how to certify that the optimal makespan is high because of the communication delays. However, this argument suffices for a
$(c+2)$-approximation, since $p_j \geq 1$ for all $j \in J$.

As pointed out by Bansal~\cite{Bansalmapsp}, there is no known promising LP relaxation. To understand the issue
let us consider the special case $\p \infty \mid \Prec, p_j=1, c \mid \Cmax$. Extending, for example, the LP of Munier and K\"onig~\cite{MunierKonig}, one might choose
variables $C_j$ as completion times, as well as decision variables $x_{j_1,j_2}$ denoting whether $j_2$ is executed in the time window $[C_{j_1},C_{j_1}+c)$ on the same machine as $j_1$. Then we can try to enforce communication delays by requiring that $C_{j_2} \geq C_{j_1} + 1 + (c-1) \cdot (1-x_{j_1,j_2})$ for $j_1 \prec j_2$. Further, we enforce load constraints $\sum_{j_1 \in J} x_{j_1,j_2} \leq c$ for $j_2 \in J$ and $\sum_{j_2 \in J} x_{j_1,j_2} \leq c$ for $j_1 \in J$. To see why this LP fails, note that in any instance where the maximum dependence degree is bounded by $c$, one could simply set $x_{j_1,j_2} = 1$ and completely avoid paying any communication delay.
Moreover, this problem seems to persist when moving to more complicated LPs that incorporate indices for time and machines.
%, the hope arises that this is indeed a \emph{local consistency issue}$
%rather than a global one.

A convenient observation is that, in exchange for a constant-factor loss in the approximation guarantee, it suffices to
find an assignment of jobs to length-$c$ intervals such that dependent jobs scheduled in the same length-$c$ interval must be
assigned to the same machine.
(The latter condition will be enough to satisfy the communication delay constraints
as, intuitively, between every two length-$c$ intervals we will insert an empty one.)
In order to obtain a stronger LP relaxation, we consider an \emph{$O(1)$-round Sherali-Adams lift} of an inital LP with indices for time and machines. %In  which is slightly non-standard as we have mix of binary and continuous decision variables.
%the underlying LP which in our setting
%can still be done with a bit of care, even though we have a mix of binary continuous decision variables.
From the lifted LP, we extract a \emph{distance function} $d: J \times J \to [0,1]$ which satisfies the following properties: 
\begin{enumerate}
\item[(i)] The function $d$ is a \emph{semimetric}.
\item[(ii)] $C_{j_1} +  d(j_1,j_2) \leq C_{j_2}$ for $j_1 \prec j_2$.
\item[(iii)] Any set $U \subseteq J$ with a diameter of at most $\frac{1}{2}$ w.r.t. $d$, satisfies $|U| \leq 2c$.
\end{enumerate}
%with the following properties: 
Here we have changed the interpretation of $C_j$ to the \emph{index} of the length-$c$ interval in which $j$ will be processed.
Intuitively, $d(j_1,j_2)$ can be understood as the probability that jobs $j_1,j_2$ are \emph{not} being scheduled within the same length-$c$ interval on the same machine. To see why a constant number of Sherali-Adams rounds are helpful, observe that the triangle inequality behind $(i)$ is really a property depending only on \emph{triples} $\{ j_1,j_2,j_3\}$ of jobs and an $O(1)$-round Sherali-Adams lift
would be locally consistent for every triple of variables.

We will now outline how to round such an LP solution. 
For jobs whose LP completion times are sufficiently different, say $C_{j_1} + \Theta(\frac{1}{\log(n)}) \leq C_{j_2}$,
we can afford to deterministically schedule $j_1$ and $j_2$ at least $c$ time units apart while only paying a $O(\log n)$-factor more than the LP. Hence the critical case is to sequence a set of jobs $J^* = \{ j \in J \mid C^* \leq C_j \leq C^* + \Theta(\frac{1}{\log(n)}) \}$
whose LP completion times are very close to each other. Note that by property $(ii)$, we know that any dependent jobs $j_1,j_2 \in J^*$ must have $d(j_1,j_2) \leq \Theta(\frac{1}{\log(n)})$.
As $d$ is a semimetric, we can make use of the rich toolset from the theory of metric spaces. In particular, we use an algorithm by Calinescu, Karloff and Rabani~\cite{DBLP:journals/siamcomp/CalinescuKR04}: For a parameter $\Delta>0$, one can partition a semimetric space into \emph{random clusters} so that the
diameter of every cluster is bounded by $\Delta$ and each $\delta$-neighborhood around a node is separated with probability at most  $O(\log(n)) \cdot \frac{\delta}{\Delta}$.
Setting $\delta := \Theta(\frac{1}{\log(n)})$ and $\Delta := \Theta(1)$ one can then show that a fixed job $j \in J^*$ will be
in the same cluster as \emph{all} its ancestors in $J^*$ with probability at least $\frac{1}{2}$, while all clusters have diameter at most $\frac{1}{2}$. By $(iii)$, each cluster will contain at most $2c$ many (unit-length) jobs, and consequently we can schedule all the clusters in parallel, where we drop any job that got separated from any ancestor.
Repeating the sampling $O(\log n)$ times then schedules all jobs in $J^*$. This reasoning results in a $O(\log^2 n)$-approximation
for this problem, which we call $\p \infty \mid \Prec, p_j=1, \Intervals \mid \Cmax$. With a bit of care the approximation factor can be improved to $O(\log c \cdot \log m)$.

Finally, the promised $O(\log c \cdot \log m)$-approximation for the more general problem $\p \mid \Prec, c \mid \Cmax$
follows from a reduction to the described special case $\p \infty \mid \Prec, p_j=1, \Intervals \mid \Cmax$.
%Finally we give a reduction that implies the same approximation factor for

\subsection{History of the Problem} \label{sec:related_work}

Precedence-constrained scheduling problems of minimizing the makespan and sum of completion times objectives have been extensively studied  for many decades in various settings. We refer the reader to~\cite{michael2018scheduling,lawler1993sequencing,PruhsST04,ambuhl2008precedence,svensson2009approximability} for more details.
Below, we only discuss results directly related  to the communication delay problem in the offline setting.

\paragraph{Approximation algorithms.}
As mentioned earlier,  Graham's~\cite{GrahamListScheduling1966} list scheduling algorithm
yields a $(c+2)$-approximation
for
$\p \mid \Prec, c \mid \Cmax$,
and a $(c+1)$-approximation for the $\p \infty$ variant.
For unit-size jobs and $c \ge 2$,
Giroudeau, K\"onig, Moulai and Palaysi~\cite{GiroudeauKMP08}
improved the latter
($\p \infty \mid \Prec, p_j=1, c \ge 2 \mid \Cmax$)
to a $\frac{2}{3}(c+1)$-approximation.
For unit-size jobs and $c = 1$,
Munier and K\"onig~\cite{MunierKonig}
obtained a $4/3$-approximation
via LP rounding % favorite successor
($\p \infty \mid \Prec, p_j=1, c=1 \mid \Cmax$);
for the $\p$ variant,
Hanen and Munier~\cite{HanenMunier73Apx}
gave an easy reduction from the $\p \infty$ variant
that loses an additive term of $1$ in the approximation ratio,
thus yielding a $7/3$-approximation.
Thurimella and Yesha~\cite{ThurimellaYesha}
gave a reduction that,
given an $\alpha$-approximation algorithm for $\p \infty \mid \Prec, c, p_j=1 \mid \Cmax$,
would yield a $(1 + 2 \alpha)$-approximation algorithm for $\p \mid \Prec, c, p_j=1 \mid \Cmax$.

For a constant number of machines,
a hierarchy-based approach of Levey and Rothvoss~\cite{LeveyR16} for the no-delay setting
($\p m \mid \Prec, p_j=1 \mid \Cmax$)
was 
generalized by Kulkarni, Li, Tarnawski and Ye~\cite{KulkarniLTY20}
to allow for communication delays that are also bounded by a constant.
For any $\varepsilon > 0$ and $\hat c \in \setZ_{\ge 0}$,
they give a nearly quasi-polynomial-time $(1+\varepsilon)$-approximation algorithm
for $\p m \mid \Prec, p_j=1, c_{jk} \le \hat c \mid \Cmax$.
The result also applies to arbitrary job sizes, under the assumption that preemption of jobs is allowed, but migration is not.

\iffalse
\subsection{Hierarchies in Scheduling}
One can take a  linear programming relaxation for any optimization problem that has a large integrality gap, 
and strengthen it automatically by applying an LP or SDP hierarchy lift. 
While this principle has been known since a long time, there are very few results in scheduling algorithms  where hierarchies have helped obtaining better algorithms. 
\fi

\paragraph{Hardness.}
Hoogeveen, Lenstra and Veltman~\cite{HoogeveenLV94} showed that even the special case $\p \infty \mid \Prec, p_j=1, c=1 \mid \Cmax$ is NP-hard to approximate to a factor  better than $7/6$.
For the case with bounded number of machines (the $\p$~variant)
%setting $\p \mid \Prec, p_j=1, c=1 \mid \Cmax$ (the number of processors being part of the input)
they show $5/4$-hardness.
These two results can be generalized for $c \ge 2$ to $(1+1/(c+4))$-hardness~\cite{GiroudeauKMP08}
and $(1+1/(c+3))$-hardness~\cite{BampisGK96},
respectively.% (the latter holds also for the $\dup$ version).
\footnote{
	Papadimitriou and Yannakakis~\cite{PapadimitriouY90} claim a $2$-hardness for
	%the no-duplication version
	$\p \infty \mid \Prec, p_j=1, c \mid \Cmax$,
	but give no proof.
	Schuurman and Woeginger~\cite{SW99a} remark that ``it would be nice to have a proof for this claim''.
}

\paragraph{Duplication model.}
The communication delay problem has also been studied (to a lesser extent) in a setting where jobs can be duplicated (replicated), i.e., executed on more than one machine, in order to avoid communication delays.
%As an example, consider an instance with $n-1$ machines, $n$ unit-size jobs with precedence constraints $1 \prec 2$, $1 \prec 3$, $1 \prec 4$, ..., $1 \prec n$, and $c = n$. The optimal makespan without duplication  is $n$, but using duplication one can get makespan $2$
%(begin by executing job $1$ on every machine).
This assumption seems to significantly simplify the problem,
especially when we are also given an unbounded number of machines:
already in 1990,
Papadimitriou and Yannakakis~\cite{PapadimitriouY90} gave a rather simple $2$-approximation algorithm for $\p \infty \mid \Prec, p_j, c_{jk}, \dup \mid \Cmax$.
Observe that this result holds even when communication delays are unrelated (they depend on the pair of jobs).
%Colin and Chr\'etienne~\cite{ColinC91}
%gave an exact algorithm
%for the case
%$\p \infty \mid \Prec, c_{jk}, \dup, \max c_{jk} \le \min p_j \mid \Cmax$,
%where processing times are arbitrary,
%but communication delays are smaller than processing times.
%Jung, Kirousis and Spirakis~\cite{JungKS93}
%gave an exact algorithm
%based on dynamic programming
%for $\p \infty \mid \Prec, p_j=1, c, \dup \mid \Cmax$
%that runs in time $O(n^{c+1})$.
%Munier and Hanen~\cite{HanenMunierDuplication}
%gave a $2$-approximation\footnote{We ignore lower-order terms such as $-1/m$.}
%for $\p \mid \Prec, p_j=1, c=1, \dup \mid \Cmax$.
The only non-trivial approximation algorithm for arbitrary $c$ and a bounded number of machines is due to Lepere and Rapine~\cite{LepereR02}, who gave an asymptotic $O(\log c / \log \log c)$-approximation for $\p \mid \Prec, p_j=1, c, \dup \mid \Cmax$.
On the hardness side,
Papadimitriou and Yannakakis~\cite{PapadimitriouY90}
showed NP-hardness of
$\p \infty \mid \Prec, p_j=1, c, \dup \mid \Cmax$
(using a large delay $c = \Theta(n^{2/3})$).

Besides being seemingly easier to approximate,
we also believe that the replication model is less applicable in most real-world scenarios due to the computation and energy cost of replication, as well as because replication is more difficult to achieve if the computations are nondeterministic in some sense (e.g.~randomized).

\bigskip
\noindent
Many further references can be found in~\cite{VeltmanLL90,GiroudeauKMP08,Drozdowski09,GiroudeauKoenig07,ColinC91,JungKS93,HanenMunierDuplication}.

\section{Preliminaries}

\subsection{The Sherali-Adams Hierarchy for LPs with Assignment Constraints}

In this section, we review the \emph{Sherali-Adams hierarchy} which provides an automatic strengthening of
linear relaxations for 0/1 optimization problems.
The authorative reference is certainly Laurent~\cite{Comparison-of-Hierarchies-Laurent-MOR03},
and we adapt the notation from Friggstad et al.~\cite{LP-for-DST-FriggstadKKLST-IPCO14}.
Consider a set of variable indices $[n] = \{ 1,\ldots,n\}$ and let $U_1,\ldots,U_N \subseteq [n]$
be subsets of variable indices.
We consider a polytope
\[
  K = \Big\{ x \in \setR^n \mid \tilde{A}x \geq \tilde{b}, \;\; \sum_{i \in U_k} x_i = 1 \;\; \forall k \in [N], \;\; 0 \leq x_i \leq 1 \;\; \forall i \in [n] \Big\},
\]
which we also write in a more compact form as $K = \{ x \in \setR^n \mid Ax \geq b\}$ with $A \in \setR^{m \times n}$ and $b \in \setR^m$.
We note that we included  explicitly the ``box constraints'' $0 \leq x_i \leq 1$ for all variables $i$.
Moreover, the constraint matrix contains \emph{assignment constraints} of the form $\sum_{i \in U_k} x_i = 1$.
This is the aspect that is non-standard in our presentation.

The general goal is to obtain a strong relaxation for the integer hull $\textrm{conv}( K \cap \{ 0,1\}^n)$.
Observe that any point $x \in \textrm{conv}( K \cap \{ 0,1\}^n)$ can be interpreted as a \emph{probability distribution} $X$ over points $K \cap \{ 0,1\}^n$. We know that any distribution can be described by the $2^n$ many values
$y_{I} = \Pr[\bigwedge_{i \in I}(X_i=1)]$ for $I \subseteq [n]$ --- in fact, the probability of any other event can be reconstructed
using the \emph{inclusion-exclusion formula}, for example $\Pr[X_1=1\textrm{ and }X_2=0] = y_{\{1\}}-y_{\{1,2\}}$.
While this is an exact approach, it is also an inefficient one.
In order to obtain a polynomial-size LP, we only work with variables $y_I$ where $|I| \leq O(1)$. 
Hence, for $r \geq 0$, we denote $\pazocal{P}_r([n]) := \{ S \subseteq [n] \mid |S| \leq r\}$ as all the index sets of size at most $r$.

\begin{definition}
  Let $SA_r(K)$ be the set of vectors $y \in \setR^{\pazocal{P}_{r+1}([n])}$ satisfying $y_{\emptyset} = 1$ and
  \[
 \sum_{H \subseteq J} (-1)^{|H|} \cdot \Big(\sum_{i=1}^n A_{\ell,i}y_{I \cup H \cup \{ i\}} - b_{\ell}y_{I \cup H}\Big) \geq 0 \quad \forall \ell \in [m]
\]
for all $I,J \subseteq [n]$ with $|I| + |J| \leq r$.
\end{definition}
The parameter $r$ in the definition is usually called the \emph{rank} or \emph{number of rounds} of the Sherali-Adams
lift.
It might be helpful for the reader to verify that for $I=J=\emptyset$, the constraint simplifies to $\sum_{i=1}^n A_{\ell,i}y_{\{ i\}} \geq b_{\ell}y_{\emptyset}=b_{\ell}$,
which implies that $(y_{\{1\}},\ldots,y_{\{n\}}) \in K$. Moreover it is instructive to verify that for any feasible integral solution  $x \in K \cap \{ 0,1\}^n$ one can set
$y_{I} := \prod_{i \in I} x_i$ to obtain a vector $y \in SA_r(K)$.

\begin{theorem}[Properties of Sherali-Adams] \label{thm:PropertiesOfSA}
  Let $y \in SA_r(K)$ for some $r \geq 0$. Then the following holds: 
  \begin{enumerate}
  \item[(a)] For $J \in \pazocal{P}_r([n])$ with $y_{J} > 0$, the vector $\tilde{y} \in \setR^{\pazocal{P}_{r+1-|J|}([n])}$ defined by $\tilde{y}_{I} := \frac{y_{I \cup J}}{y_J}$ satisfies $\tilde{y} \in SA_{r-|J|}(K)$.
  \item[(b)] One has $0 \leq y_{I} \leq y_J \leq 1$ for $J \subseteq I$ and $|I| \leq r+1$.
  \item[(c)] If $|J| \leq r+1$ and $y_i \in \{ 0,1\} \; \forall i \in J$, then $y_I = y_{I \setminus J} \cdot \prod_{i \in I \cap J} y_i$ for all $|I| \leq r+1$.
  \item[(d)] For $J \subseteq [n]$ with $|J| \leq r$ there exists a distribution over vectors $\tilde{y}$ such that $(i)$ $\tilde{y} \in SA_{r-|J|}(K)$, (ii) $\tilde{y}_i \in \{ 0,1\}$ for $i \in J$, (iii) $y_I = \E[\tilde{y}_I]$ for all $I \subseteq [n]$ with $|I \cup J| \leq r+1$ (this includes in particular all $I \in \pazocal{P}_{r+1-|J|}([n])$).
    \item[(e)] For $I \subseteq [n]$ with $|I| \leq r$ and $k \in [N]$ one has $y_I = \sum_{i \in U_k}y_{I \cup \{i\}}$.
  \item[(f)] Take $H \subseteq [N]$ with $|H| \leq r$ and set $J := \bigcup_{k \in H} U_k$. Then there exists a distribution over vectors $\tilde{y}$ such that (i) $\tilde{y} \in SA_{r-|H|}(K)$, (ii) $\tilde{y}_i \in \{ 0,1\}$ for $i \in J$, (iii) $y_I = \E[\tilde{y}_I]$ for all $I \in \pazocal{P}_{r+1-|H|}([n])$.
  \end{enumerate}
  
\end{theorem}
\begin{proof}
  For (a)-(d), we refer to the extensive coverage in Laurent~\cite{Comparison-of-Hierarchies-Laurent-MOR03}. We prove (e) and (f) which are non-standard and custom-tailored to LPs with assignment constraints: 
  \begin{enumerate}
  \item [(e)] Fix $I \subseteq [n]$ with $|I| \leq r$. We apply (d) to obtain a distribution over $\tilde{y}$ with $\tilde{y} \in SA_{r-|I|}(K)$ so that $\tilde{y}_i \in \{ 0,1\}$ for $i \in I$. Then
    \[
\sum_{i \in U_k} y_{I \cup \{ i\}}\stackrel{\textrm{linearity}}{=} \E\Big[ \sum_{i \in U_k} \tilde{y}_{I \cup \{ i\}}\Big] \stackrel{(c)}{=} \E\Big[ \tilde{y}_I \cdot \underbrace{\sum_{i \in U_k} \tilde{y}_i}_{=1}\Big] = \E[\tilde{y}_I] = y_I.
    \]
Here we apply $(c)$ for index sets $I \cup \{ i\}$ where variables in $J := I$ have been made integral. Note that indeed $|I \cup (I \cup \{ i\})| \leq r+1$ as required. 
\item[(f)] By an inductive argument it suffices to consider the case of $|H| = 1$.
  Let $H = \{ k\}$ and set $U := U_k$, i.e. the constraints for polytope $P$ contain the assignment constraint $\sum_{i \in U} x_i = 1$ and we want to make all variables
  in $U$ integral while only losing a \emph{single} round in the hierarchy. Abbreviate $U^+ := \{ i \in U \mid y_{\{i\}} > 0\}$.
  For $i \in U^+$, define $y^{(i)} \in \setR^{\pazocal{P}_{r}([n])}$
  to be the vector with $y^{(i)}_I := \frac{y_{I \cup \{ i\}}}{y_i}$. By (a) we know that $y^{(i)} \in SA_{r-1}(K)$. Moreover $y^{(i)}_{\{i\}} = \frac{y_{\{i\}}}{y_{\{i\}}} = 1$. Then the assignment constraint of the LP forces that $y_{\{i'\}}^{(i)} = 0$ for $i' \in U \setminus \{ i\}$. Now we define a probability distribution over vectors
  $\tilde{y}$ as follows: for $i \in U^+$, with probability $y_i$ we set $\tilde{y} := y^{(i)}$. Then (i) and (ii) hold for $\tilde{y}$  as discussed.
  Property (iii) follows from
  \[
\E[\tilde{y}_I] = \sum_{i \in U^+} y_i y_I^{(i)} = \sum_{i \in U^+} y_i \frac{y_{I \cup \{i\}}}{y_i} = \sum_{i \in U^+} y_{I \cup \{i\}} \stackrel{(b)}{=} \sum_{i \in U} y_{I \cup \{ i\}} \stackrel{(e)}{=} y_I
  \]
 \end{enumerate} 
\end{proof}
It is known that Theorem~\ref{thm:PropertiesOfSA}.(f) holds in a stronger form for the SDP-based \emph{Lasserre hierarchy}.
 Karlin, Mathieu and Nguyen~\cite{Hierarchies-for-Knapsack-KarlinMathieuNguyen-IPCO11} proved a result that can be paraphrased as follows:
  \emph{if one has any set $J \subseteq [n]$ of variables with the property that
there is no LP solution with more than $k$ ones in $J$, then one can make all variables of $J$ integral while losing only $k$ rounds}.
Interestingly, Karlin, Mathieu and Nguyen prove that this is completely false for Sherali-Adams. 
In particular, for a Knapsack instance with unit size items and capacity $2-\varepsilon$, 
the integrality gap is still $2-2\varepsilon$ after $\Theta_{\varepsilon}(n)$ rounds of Sherali-Adams.
In a different setting, Friggstad et al.~\cite{LP-for-DST-FriggstadKKLST-IPCO14}
realized that given a ``tree constraint'', a Sherali-Adams lift can provide
the same guarantees that Rothvoss~\cite{DirectedSteinerTreeAndLasserre-RothvossArxiv2011} derived from Lasserre. 
While Friggstad et al.~did not state their insight in the generality that we need here, our Lemma~\ref{thm:PropertiesOfSA}.(e)+(f) are inspired by their work.

\subsection{Semimetric Spaces}
\label{sec:semimetric spaces}

A \emph{semimetric space} is a pair $(V,d)$ where $V$ is a finite set (we denote $n := |V|$)
and $d : V \times V \to \setR_{\geq 0}$ is a \emph{semimetric}, i.e.
\begin{itemize}
\item $d(u,u) = 0$ for all $u \in U$.
\item Symmetry: $d(u,v) = d(v,u)$ for all $u,v \in V$.
\item Triangle inequality: $d(u,v) + d(v,w) \geq d(u,w)$ for all $u,v,w \in V$. 
\end{itemize}
Recall that the more common notion is that of a \emph{metric}, which additionally requires
that $d(u,v) > 0$ for $u \neq v$.
For a set $U \subseteq V$ we denote the \emph{diameter} as $\textrm{diam}(U) := \max_{u,v \in U} d(u,v)$. 
Our goal is to find a partition $V = V_1 \dot{\cup} \ldots \dot{\cup} V_k$ such that the diameter
of every cluster $V_i$ is bounded by some parameter $\Delta$.
We denote $d(w,U) := \min\{ d(w,u) : u \in U\}$ as the distance to the set $U$.
Moreover, for $r \geq 0$ and $U \subseteq V$, let $N(U,r) := \{ v \in V \mid d(v,U) \leq r\}$ 
be the \emph{distance $r$-neighborhood} of $U$.

We use a very influential clustering algorithm due to
Calinescu, Karloff and Rabani~\cite{DBLP:journals/siamcomp/CalinescuKR04},
which assigns each $v \in V$ to a random cluster center $c \in V$ such that $d(u,c) \leq \beta \Delta$.
Nodes assigned to the same cluster center form one block $V_i$ in the partition. 
Formally the algorithm is as follows:
\begin{center}
  \psframebox{
  \begin{minipage}{14cm}
\textsc{CKR Clustering algorithm} \vspace{2mm} \hrule \vspace{1mm}
{\bf Input:} Semimetric space $(V,d)$ with  $V = \{ v_1,\ldots,v_n\}$, parameter $\Delta>0$ \\
{\bf Output:} Clustering $V = V_1 \dot{\cup} \ldots \dot{\cup} V_k$ for some $k$. \vspace{1mm} \hrule \vspace{1mm}
\begin{enumerate*}\label{alg: CKR}
\item[(1)] Pick a uniform random $\beta \in [\frac{1}{4},\frac{1}{2}]$
\item[(2)] Pick a random ordering $\pi : V \to \{ 1,\ldots,n\}$ 
\item[(3)] For each $v \in V$ set $\sigma(v) := v_{\ell}$ so that $d(v,v_{\ell}) \leq \beta \cdot \Delta$ and $\pi(v_\ell)$ is minimal 
\item[(4)] Denote the points $v \in V$ with $\sigma^{-1}(v) \neq \emptyset$ by $c_1,\ldots,c_k \in V$ and return clusters $V_i := \sigma^{-1}(c_i)$ for $i=1,\ldots,k$
\end{enumerate*}
\end{minipage}}
\end{center}

Note that the algorithm has two sources of randomness: it picks a random parameter $\beta$,
and independently it picks a random ordering $\pi$. Here the ordering is to be understood such that 
element $v_{\ell}$ with $\pi(v_{\ell}) = 1$ is the ``highest priority'' element. 
%We will state and reprove
%the analysis of the CKR Clustering Algorithm with a slight generalization from pairs (i.e. $|U|=2$) to the
%separation of a general node set.
The original work of Calinescu, Karloff and Rabani~\cite{DBLP:journals/siamcomp/CalinescuKR04} only provided an upper bound on the probability that a short edge $(u,v)$ is separated. Mendel and Naor~\cite{RamseyPartitions-MendelNaorFOCS06} note that the same clustering provides the guarantee of
$\Pr[N(u,t)\textrm{ separated}] \leq 1 - O(\frac{t}{\Delta} \cdot \ln( \frac{|N(u,\Delta)|}{|N(u,\Delta/8)|}))$
for all $u \in V$ and $0\leq t<\frac{\Delta}{8}$.
% where $B(u,t) := \{ v \in V \mid d(u,v) \leq t\}$ is the radius-$t$ ball around $u$ in the space.
Mendel and Naor attribute this
to Fakcharoenphol, Rao and Talwar~\cite{TreeMetricFakcharoenpholRaoTalwar-JCSS04} (while Fakcharoenphol, Rao and Talwar\cite{TreeMetricFakcharoenpholRaoTalwar-JCSS04} do not state it explicitly in this form and focus on the ``local growth ratio'' aspect).

%However it is not hard to observe that a minor modification of the analysis also upperbounds the probability for a low-diameter set to be separated.
%The statement that we require is implied by Remark 3.1. in Mendel and Naor~? who attribute the statement to \footnote{The paper of FKR }
% We will show the following properties of the CKR Clustering Algorithm.
% . by Calinescu, Karloff and Rabani.
We state the formal claim in a form that will be convenient for us.
For the sake of completeness, a proof can be found in the Appendix.
\begin{theorem}[Analysis of CKR] \label{thm:ProbUSeperatedByClustering}
  Let $V = V_1 \dot{\cup} \ldots \dot{\cup} V_k$ be the random partition of the CKR algorithm. The following holds:
  \begin{enumerate*}
  \item[(a)] The blocks have $\textrm{diam}(V_i) \leq \Delta$ for $i=1,\ldots,k$.
  \item[(b)] Let $U \subseteq V$ be a subset of points. Then \vspace{-2mm}
  \[
  \Pr[U\textrm{ is separated by clustering}] \leq \ln\big(2\big|N\big(U,\Delta/2 \big)\big|\big) \cdot \frac{4\textrm{diam}(U)}{\Delta} \leq \ln(2n) \cdot \frac{4\textrm{diam}(U)}{\Delta}.
\]
\end{enumerate*}
\end{theorem}
In the above, \emph{separated} means that there is more than one index $i$ with $V_i \cap U \neq \emptyset$.
%There are metric spaces (e.g. take $d(i,j) := |i-j|$ with $V = \{ 1,\ldots,n\}$) where any partition will separate two points that are close.
%To overcome this,
%we use a random partition, so that close points only have a
%a small probability of being separated.

\section{An Approximation for  $\p \infty \mid \Prec, p_j=1, \Intervals \mid C_{\max}$}
\label{sec:approximation_for_pinfty}

In this section, we provide an approximation algorithm for scheduling $n$ unit-length jobs $J$ with
communication delay $c \in \setN$ on an unbounded number of machines so that precedence constraints given
by a partial order $\prec$ are satisfied. Instead of working with  $\p \infty \mid \Prec, p_j=1, c \mid C_{\max}$ directly, it will be
more convenient to consider a slight variant that we call $\p \infty \mid \Prec, p_j=1, \Intervals \mid C_{\max}$.
This problem variant has the same input but the time horizon is partitioned into time {\em intervals} of length $c$,
say $I_s = [sc, (s+1)c)$ for $s \in \setZ_{\geq 0}$.
The goal is to assign jobs to intervals and machines. % so that precedence constraints are preserved.
We require that if $j_1 \prec j_2$
then either $j_1$ is scheduled in an earlier interval than $j_2$ or $j_1$ and $j_2$ are scheduled in the same
interval on the same machine. Other than that, there are no further communication delays.
The objective function is to minimize the number of intervals used to process the jobs.
In fact we do not need to decide the order of jobs within intervals as any topological order will work.
In a more mathematical notation, the problem asks to find a partition $J = \dot{\bigcup}_{s \in \{0,\ldots,S-1\},i \in \setN} J_{s,i}$
with $|J_{s,i}| \leq c$ such that $S$ is minimized and for every $j_1 \prec j_2$ with $j_1 \in J_{s_1,i_1}$ and  $j_2 \in J_{s_2,i_2}$
one has either $s_1<s_2$ or $(s_1,i_1)=(s_2,i_2)$.
%as we can simply allow idle time $c$ between intervals.
See Figure~\ref{fig:ExamplePinftyIntervals} for an illustration.

It is rather straightforward to give reductions between $\p \infty \mid \Prec, p_j=1, c \mid C_{\max}$ and $\p \infty \mid \Prec, p_j=1, \Intervals \mid C_{\max}$ that only lose a small constant factor in both directions. The only subtle point to consider here
is that when the optimum makespan for $\p \infty \mid \Prec, c \mid C_{\max}$ is less than $c$, the problem admits a PTAS;
we refer to Section~\ref{sec:Reductions} for details.

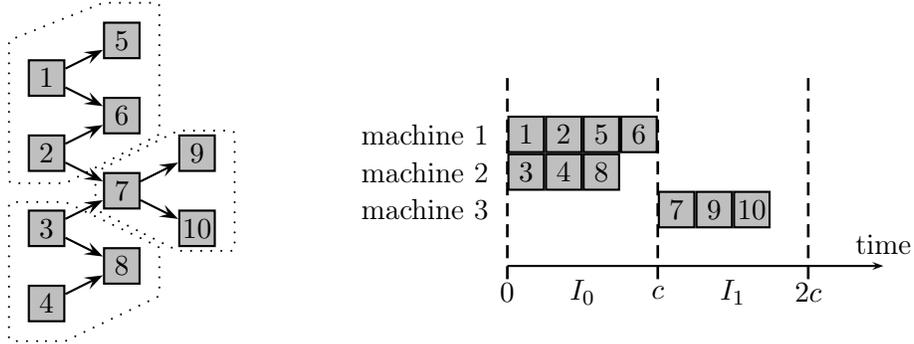
\begin{figure}
\begin{center}
  \psset{unit=0.5cm}
  \begin{pspicture}(0,0)(6,6)
  %  \psaxes(0,0)(-1,-1)(5,7)
    \pspolygon[linestyle=dotted,linearc=0.1](-1,3.2)(0.5,3.2)(3,4.5)(3,8)(1.5,8)(-1,7)
    \pspolygon[linestyle=dotted,linearc=0.1](-1,-1)(-1,2.7)(0.5,2.7)(3,1.5)(3,0.2)(0.5,-1)
    \pspolygon[linestyle=dotted,linearc=0.1](1.3,2.5)(1.3,3.5)(3.5,4.6)(5,4.6)(5,1.4)(3.5,1.4)
    \fnode[framesize=1,fillcolor=lightgray,fillstyle=solid](0,6){j1}\rput[c](j1){$1$}
    \fnode[framesize=1,fillcolor=lightgray,fillstyle=solid](0,4){j2}\rput[c](j2){$2$}
    \fnode[framesize=1,fillcolor=lightgray,fillstyle=solid](0,2){j3}\rput[c](j3){$3$}
    \fnode[framesize=1,fillcolor=lightgray,fillstyle=solid](0,0){j4}\rput[c](j4){$4$}
    \fnode[framesize=1,fillcolor=lightgray,fillstyle=solid](2,7){j5}\rput[c](j5){$5$}
    \fnode[framesize=1,fillcolor=lightgray,fillstyle=solid](2,5){j6}\rput[c](j6){$6$}
    \fnode[framesize=1,fillcolor=lightgray,fillstyle=solid](2,3){j7}\rput[c](j7){$7$}
    \fnode[framesize=1,fillcolor=lightgray,fillstyle=solid](2,1){j8}\rput[c](j8){$8$}
    \fnode[framesize=1,fillcolor=lightgray,fillstyle=solid](4,4){j9}\rput[c](j9){$9$}
    \fnode[framesize=1,fillcolor=lightgray,fillstyle=solid](4,2){j10}\rput[c](j10){$10$}
    \ncline[arrowsize=5pt]{->}{j1}{j5}
    \ncline[arrowsize=5pt]{->}{j1}{j6}
    \ncline[arrowsize=5pt]{->}{j2}{j6}
    \ncline[arrowsize=5pt]{->}{j2}{j7}
    \ncline[arrowsize=5pt]{->}{j3}{j7}
    \ncline[arrowsize=5pt]{->}{j3}{j8}
    \ncline[arrowsize=5pt]{->}{j4}{j8}
    \ncline[arrowsize=5pt]{->}{j7}{j9}
    \ncline[arrowsize=5pt]{->}{j7}{j10}
  \end{pspicture}
  \begin{pspicture}(-6,-1)(10,4)
    \fnode[framesize=1,fillcolor=lightgray,fillstyle=solid](0.5,3.5){j1}\rput[c](j1){$1$}
    \fnode[framesize=1,fillcolor=lightgray,fillstyle=solid](1.5,3.5){j2}\rput[c](j2){$2$}
    \fnode[framesize=1,fillcolor=lightgray,fillstyle=solid](2.5,3.5){j5}\rput[c](j5){$5$}
    \fnode[framesize=1,fillcolor=lightgray,fillstyle=solid](3.5,3.5){j6}\rput[c](j6){$6$}
    \fnode[framesize=1,fillcolor=lightgray,fillstyle=solid](0.5,2.5){j3}\rput[c](j3){$3$}
    \fnode[framesize=1,fillcolor=lightgray,fillstyle=solid](1.5,2.5){j4}\rput[c](j4){$4$}
    \fnode[framesize=1,fillcolor=lightgray,fillstyle=solid](2.5,2.5){j8}\rput[c](j8){$8$}
    \fnode[framesize=1,fillcolor=lightgray,fillstyle=solid](4.5,1.5){j7}\rput[c](j7){$7$}
    \fnode[framesize=1,fillcolor=lightgray,fillstyle=solid](5.5,1.5){j9}\rput[c](j9){$9$}
     \fnode[framesize=1,fillcolor=lightgray,fillstyle=solid](6.5,1.5){j10}\rput[c](j10){$10$}
    \psline[linewidth=1pt,linestyle=dashed](0,-5pt)(0,5)
    \psline[linewidth=1pt,linestyle=dashed](4,-5pt)(4,5)
    \psline[linewidth=1pt,linestyle=dashed](8,-5pt)(8,5)
    \psline{->}(0,0)(10,0) \rput[c](10,8pt){time}
    \rput[r](-0.5,3.5){machine 1}
    \rput[r](-0.5,2.5){machine 2}
    \rput[r](-0.5,1.5){machine 3}
    \rput[c](0,-10pt){$0$}
    \rput[c](4,-10pt){$c$}
    \rput[c](8,-10pt){$2c$}
    \rput[c](2,-10pt){$I_0$}
    \rput[c](6,-10pt){$I_1$}
  \end{pspicture}
\caption{Left: example of an instance of $\p \infty \mid \Prec, p_j=1, \Intervals \mid C_{\max}$ with $c=4$ (where the partial order $\prec$ is the transitive closure of the depicted digraph). Right: a valid schedule in 2 intervals.\label{fig:ExamplePinftyIntervals}}  
\end{center}
\end{figure}
%In the next section, we give a reduction from $\p m \mid \Prec \mid \Cmax$ to this special case $P\infty \mid \Prec, p_j=1, c\textrm{-intervals} \mid C_{\max}$ losing only $O(1)$-factor  in the approximation ratio.

\subsection{The Linear Program}

Let $m \in \setN$ be a parameter defining the number of machines that we admit for the LP.  
Moreover, let $S \in \setN$ be the number of intervals that we allow for the time horizon.
To obtain an approximation for $\p \infty \mid \Prec, p_j=1, \Intervals \mid C_{\max}$  one can set $m := n$ and perform a binary
search to find the minimal $S$ for which the LP is feasible. But we prefer to keep the approach general.

%By doing  binary search, we can assume that an optimal solution to a given instance can be scheduled using at most $S$ intervals.
We construct the LP in two steps. % for $P\infty \mid \Prec, p_j=1, c\textrm{-intervals} \mid C_{\max}$ in two steps. 
First consider the variables
\[
  x_{j,i,s} = \begin{cases} 1 & \textrm{if }j\textrm{ is scheduled on machine }i\textrm{ in interval }I_s \\ 0 & \textrm{otherwise} \end{cases} \quad \forall j \in J, i \in [m], s \in \{ 0,\ldots,S-1\}
%  \theta_j = index of interval where j is scheduled 
\]
Let $K$ be the set of fractional solutions to the following linear system
\begin{eqnarray*}
  \sum_{i \in [m]} \sum_{s \geq 0} x_{j,i,s} &=& 1 \quad \forall j \in J \\
  \sum_{j \in J} x_{j,i,s} &\leq& c \quad \forall i \in [m] \;\; \forall s \in \{ 0,\ldots,S-1\} \\
%  \sum_{s'\leq s} \sum_{i \in [m]} x_{j_1,i,s'} &\geq& \sum_{s' \leq s} \sum_{i \in [m]} x_{j_2,i,s'} \quad \forall j_1 \prec j_2 \;\; \forall s \in \{ 0,\ldots,S-1\} \\
  0 \leq x_{j,i,s} &\leq& 1 \quad \forall j \in J, i \in [m], s \in \{ 0,\ldots,S-1\}
\end{eqnarray*}
So far, $K$ simply assigns jobs (fractionally) to intervals and machines without taking any precedence  constraints into account. %is consists  assignment 
%The LP asks for a (fractional) schedule where for a pair $j_1 \prec j_2$, $j_2$ cannot be scheduled in an earlier interval.
%However, it does not enforce that $j_1$ and $j_2$ be performed in separate intervals if they are not scheduled on the same machine.
Next, we will use a lift $x \in \textsc{SA}_r(K)$ containing variables
$x_{(j_1,i_1,s_1),(j_2,i_2,s_2)}$, which provide the probability for the event that $j_1$ is scheduled in interval $s_1$ on machine $i_1$ and
$j_2$ is scheduled in interval $s_2$ on machine $i_2$. We introduce two more types of decision variables: 
\begin{eqnarray*}
  y_{j_1,j_2} &=& \begin{cases} 1 & j_1\textrm{ and }j_2\textrm{ are scheduled on the same machine in the same interval} \\ 0 & \textrm{otherwise} \end{cases} \\
   C_j &=& \textrm{index of interval where }j\textrm{ is processed}
\end{eqnarray*}
Let $Q(r)$ be the set of vectors $(x,y,C)$ that satisfy
\begin{eqnarray*}
  y_{j_1,j_2} &=& \sum_{s \in \{ 0,\ldots,S-1\}} \sum_{i \in [m]} x_{(j_1,i,s),(j_2,i,s)} \\
  C_{j_2} &\geq& C_{j_1} + (1-y_{j_1,j_2}) \quad \forall j_1 \prec j_2 \\
  C_j &\geq& 0 \quad \forall j \in J \\ 
  x &\in& SA_r(K)
\end{eqnarray*}
The analysis of our algorithm will work for all $r \geq 5$ while solving the LP takes time $n^{O(r)}$.
Here we make no attempt at optimizing the constant $r$.
The main technical contribution of this section is the following rounding result: 
\begin{theorem} \label{thm:LPRoundingTheorem}
Consider an instance with unit-length jobs $J$, a partial order $\prec$, and parameters $c,S,m \in \setN$ such that
  $Q(r)$ is feasible for $r:=5$. Then there is a randomized algorithm with expected polynomial running time that finds a
   schedule for $\p \infty \mid \Prec, p_j=1, \Intervals \mid C_{\max}$ using at most $O(\log m \cdot \log c) \cdot S$ intervals.
\end{theorem}
We would like to emphasize that we require $\prec$ to be a partial order, which implies that it
is transitive. While replacing any acyclic digraph with its transitive closure  does not change the set of feasible integral
schedules and hence can be done in a preprocessing step, it corresponds to adding constraints to the LP that we
rely on in the algorithm and in its analysis.
%In order for the LP to make syntactically sense we require $r \geq 2$; we will prove later that certain properties hold
%for $r \geq 5$ rounds, while we make no attempt at optimizing those constants.
%Note that the variables $y_{j_1,j_2,s}$
% Let $Q_r$ be the set of solutions

We will now discuss some properties that are implied by the Sherali-Adams lift: % gives us the following: 
\begin{lemma} \label{lem:PropertiesOfSAforSchedulingWithCommDelaysLP}
  Let $(x,y,C) \in Q(r)$ with $r \geq 2$. Then for any set $\tilde{J} \subseteq J$ of $|\tilde{J}| \leq r-2$ jobs,
  there exists a distribution $\pazocal{D}(\tilde{J})$ over pairs $(\tilde{x},\tilde{y})$ such that
  \begin{enumerate*}
  \item[(A)] $\tilde{x}_{j,i,s} \in \{ 0,1\}$ for all $j \in \tilde{J}$, all $i \in [m]$ and $s \geq 0$.
  \item[(B)]  $\tilde{y}_{j_1,j_2} = \sum_{s \geq 0} \sum_{i \in [m]} \tilde{x}_{j_1,i,s} \cdot \tilde{x}_{j_2,i,s}$ if $|\{j_1,j_2\} \cap \tilde{J}| \geq 1$.
  \item[(C)] $\tilde{x} \in K$, $\tilde{y}_{j_1,j_2} = \sum_{s \in \{ 0,\ldots,S-1\}} \sum_{i \in [m]} \tilde{x}_{(j_1,i,s),(j_2,i,s)}$
    for all $j_1,j_2 \in J$.
   \item[(D)] $\E[\tilde{x}_{j,i,s}] = x_{j,i,s}$ and $\E[\tilde{y}_{j_1,j_2}] = y_{j_1,j_2}$ for all $j,j_1,j_2,i,s$.
  \end{enumerate*}
\end{lemma}
\begin{proof}
By Theorem~\ref{thm:PropertiesOfSA}.(f), there is a distribution over $\tilde{x} \in SA_{2}(K)$
which satisfies $(A)$ and has $\tilde{x} \in K$, $\E[\tilde{x}_{j,i,s}] = x_{j,i,s}$ and 
$\E[\tilde{x}_{(j_1,i_1,s_1),(j_2,i_2,s_2)}] = x_{(j_1,i_1,s_1),(j_2,i_2,s_2)}$,
and additionally is integral on variables involving only jobs from $\tilde{J}$, where $|\tilde{J}| \leq r-2$.
Here, we crucially use that every job $j \in \tilde{J}$ is part of an assignment constraint $\sum_{i \in [m]} \sum_{s \geq 0} x_{j,i,s} = 1$,
hence making these variables integral results in the loss of only one round per job.
Then, the $y$-variables are just linear functions depending on the $x$-variables, so we can define
  \[
\tilde{y}_{j_1,j_2} := \sum_{s \in \{ 0,\ldots,S-1\}} \sum_{i \in [m]} \tilde{x}_{(j_1,i,s),(j_2,i,s)} 
  \]
  and the claim follows. 
\end{proof}

From the LP solution, we define a semimetric $d$. Here the intuitive interpretation is that
a small distance $d(j_1,j_2)$ means that the LP schedules $j_1$ and $j_2$ mostly on the same machine and in the same interval. 
%Jobs close in distance with respect to that metric are mostly scheduled on the same machine, in the same interval by the LP.
\begin{lemma}\label{lem:metric}
Let $(x,y,C) \in Q(r)$ be a solution to the LP with $r \geq 5$. Then $d(j_1,j_2) := 1-y_{j_1,j_2}$ is a semimetric.
\end{lemma}
\begin{proof}
  The first two properties from the definition of a semimetric (see Section~\ref{sec:semimetric spaces}) are clearly satisfied. We verify the triangle inequality.
  Consider three jobs $j_1,j_2,j_3 \in J$. We apply Lemma~\ref{lem:PropertiesOfSAforSchedulingWithCommDelaysLP} with $\tilde{J} := \{ j_1,j_2,j_3\}$
  and consider the distribution $(\tilde{x},\tilde{y}) \sim \pazocal{D}(\tilde{J})$.
  For $j \in \tilde{J}$, define $Z(j) = (\tilde{s}(j),\tilde{i}(s))$ as the random variable that gives the unique pair of indices such that $\tilde{x}_{j,\tilde{i}(j),\tilde{s}(j)}=1$.
  Then for $j',j'' \in \tilde{J}$ one has
  \[
   d(j',j'') = \Pr[Z(j') \neq Z(j'')] = \Pr\big[ \big(\tilde{s}(j),\tilde{i}(j')\big) \neq \big(\tilde{s}(j''),\tilde{i}(j'')\big)\big]
 \]
  Then indeed
  \begin{eqnarray*}
    d(j_1,j_3) &=& \Pr[Z(j_1) \neq Z(j_3)] \leq \Pr[Z(j_1) \neq Z(j_2) \vee Z(j_2) \neq Z(j_3)] \\
                   &\stackrel{\textrm{union bound}}{\leq}& \Pr[Z(j_1) \neq Z(j_2)] + \Pr[Z(j_2) \neq Z(j_3)] = d(j_1,j_2) + d(j_2,j_3).
 \end{eqnarray*}
\end{proof}

\begin{lemma}\label{lem:clustercapacity}
For every $j_1 \in J$ one has $\sum_{j_2 \in J} y_{j_1,j_2} \leq c$.
\end{lemma}
\begin{proof}
  Consider the distribution $(\tilde{x},\tilde{y}) \sim \pazocal{D}(\{ j_1\})$.
  From Lemma~\ref{lem:PropertiesOfSAforSchedulingWithCommDelaysLP}.(B)
  we know that $\E[\tilde{y}_{j_1,j_2}] = y_{j_1,j_2}$
  and $\tilde{y}_{j_1,j_2} = \sum_{s \in \{0,\ldots,S-1\}}\sum_{i \in [m]} \tilde{x}_{j_1,i,s} \cdot \tilde{x}_{j_2,i,s}$.
  By linearity it suffices to prove that $\sum_{j_2 \in J} \tilde{y}_{j_1,j_2} \leq c$ always.
  Fix a pair $(\tilde{x},\tilde{y})$.
  There is a unique pair of indices $(i_1,s_1)$ with $\tilde{x}_{j_1,i_1,s_1}=1$. Then 
  \[
    \sum_{j_2 \in J}\tilde{y}_{j_1,j_2} = \sum_{s \in \{ 0,\ldots,S-1\}}\sum_{j_2 \in J} \sum_{i \in [m]} \underbrace{\tilde{x}_{j_1,i,s}}_{0\textrm{ if }i \neq i_1\textrm{ or }s \neq s_1} \cdot \tilde{x}_{j_2,i,s}
    = \sum_{j_2 \in J} \tilde{x}_{j_2,i_1,s_1} \leq c.
  \] 
\end{proof}

A crucial insight is that for any job $j^*$, few jobs are very close to $j^*$ with respect to $d$.

\begin{lemma} \label{lem:SmallNeighborhoodOfClosePoints}
Fix $j^* \in J$ and abbreviate $U := \{ j \in J \mid d(j,j^*) \leq \beta\}$ for $0<\beta<1$. Then $|U| \leq \frac{c}{1-\beta}$. 
\end{lemma}
\begin{proof}
  For each $j \in U$ we have $y_{j,j^*} = 1 - d(j,j^*) \geq 1-\beta$. Combining with the last
  lemma we have $(1-\beta)|U|\leq\sum_{j \in J} y_{j,j^*} \leq c$.
\end{proof}

\subsection{Scheduling a Single Batch of Jobs}
\label{subsec:singlebatch}

We now come to the  main building block of our algorithm. We consider a subset $J^*$ of
jobs whose LP completion times $C_j$ are very close (within a $\Theta(\frac{1}{\log(c)})$ term of each other)
and show we can schedule half of these jobs in a single length-$2c$ interval. 
The following lemma is the main technical contribution of the paper.

\begin{lemma} \label{lem:SchedulingOneIntervalViaCKR}
  Let $(x,y,C) \in Q(r)$ with $r \geq 5$ and let $0<\delta \leq \frac{1}{64\log(4c)}$ be a parameter. Let $C^* \geq 0$ and set
   $J^* \subseteq \{ j \in J \mid C^* \leq C_j \leq C^*+\delta\}$.
  Then there is a randomized rounding procedure that finds a schedule for a subset $J^{**} \subseteq J^*$ in a single interval of length at most $2c$  such that every job $j \in J^*$
  is scheduled with probability at least $1 - 32 \log(4 c) \cdot \delta \geq\frac{1}{2}$.
\end{lemma}
%Note that there is always an optimum integral solution where all jobs in $J_0$ are scheduled in the first interval (as there is no reason to wait for them). We will make our life
%marginally simpler and assume that also the LP solution has $\sum_{i \in [m]} x_{jis} = 1$ for all $j \in J_0$.

%Since $p_j=1$ and $c \in \mathbb{Z}_{\geq 0}$, the makespan of a schedule is the 
%number of unit \emph{time slots} it uses, including idle time.
%We assume w.l.o.g. that the precedence order $\prec$ is transitive.
We denote $\Gamma^{-}(j)$ as the predecessors of $j$ and $\Gamma^+(j)$ as the successors, 
and similary $\Gamma^{-/+}(J') = \{ j \in J : \exists j' \in J' \textrm{ s.t. } j \in \Gamma^{-/+}(j') \}$. Again, recall that we assume $\prec$ to be transitive.
The rounding algorithm is the following:
\begin{center}
\psframebox{
\begin{minipage}{14cm}
  \textsc{Scheduling a Single Batch} \vspace{1mm} \hrule \vspace{1mm}
\begin{enumerate*}
  \item[(1)] Run a CKR clustering on the semimetric space $(J^*,d)$ with parameter $\Delta := \frac{1}{4}$ and let $V_1,\ldots,V_k$ be the clusters.
  \item[(2)] Let $V_{\ell}' := \{ j \in V_{\ell} \mid \Gamma^-(j) \cap J^* \subseteq V_{\ell}\}$ for $\ell = 1,\ldots,k$.
  \item[(3)] Schedule $V_{\ell}'$ on one machine for all $\ell=1,\ldots,k$.
  \end{enumerate*}
  \end{minipage}}
\end{center}

\begin{figure}
\begin{center}
\ifrenderfigures
  \begin{pspicture}(0,0)(3,4)
    \pspolygon[linewidth=0.4pt,linearc=0.1,fillstyle=solid,fillcolor=lightgray](-0.4,-0.4)(2.4,-0.4)(2.4,2.4)(-0.4,2.4)% V_1
    \pspolygon[linewidth=0.4pt,linearc=0.1,fillstyle=solid,fillcolor=gray](-0.3,-0.3)(2.3,-0.3)(2.3,1.3)(-0.3,1.3)% V_1'
    \pspolygon[linewidth=0.4pt,linearc=0.1,fillstyle=solid,fillcolor=lightgray](-0.4,2.6)(2.4,2.6)(2.4,5.4)(-0.4,5.4)% V_2
    \pspolygon[linewidth=0.4pt,linearc=0.1,fillstyle=solid,fillcolor=gray](-0.3,2.7)(2.3,3.7)(2.3,5.3)(-0.3,5.3)% V_2'
    \cnode*(0,1){2.5pt}{a}
    \cnode*(2,0){2.5pt}{b}
    \cnode*(1,2){2.5pt}{c}
    \cnode*(2,2){2.5pt}{d}
    \cnode*(0,3){2.5pt}{e}
    \cnode*(2,3){2.5pt}{f}
    \cnode*(1,4){2.5pt}{g}
    \cnode*(2,4){2.5pt}{h}
    \cnode*(0,5){2.5pt}{i}
    \cnode*(1,5){2.5pt}{j}
    \cnode*(2,5){2.5pt}{k}
    \ncline[arrowsize=5pt]{<-}{b}{a}
    \ncline[arrowsize=5pt]{<-}{d}{c}
 %   \ncline{->}{d}{a}
    \ncline[arrowsize=5pt]{<-}{c}{a}
 %   \ncline{->}{f}{e}
    \ncline[arrowsize=5pt]{<-}{f}{c}
    \ncline[arrowsize=5pt]{<-}{g}{e}
   \ncline[arrowsize=5pt]{<-}{h}{g}
   \ncline[arrowsize=5pt]{<-}{g}{i}
 %  \ncline{->}{h}{i}
   \ncline[arrowsize=5pt]{<-}{k}{j}
    \ncline[arrowsize=5pt]{<-}{j}{i}
 %   \ncarc[arcangle=-20]{->}{f}{a}
%    \ncarc[arcangle=-20]{->}{k}{i}
%    \ncline{->}{d}{e}
%   \ncline{->}{h}{e}
    \ncline[arrowsize=5pt]{<-}{c}{e}
%        \nput{0}{a}{a}\nput{0}{b}{b}\nput{0}{c}{c}\nput{0}{d}{d}\nput{0}{e}{e}\nput{0}{f}{f}\nput{90}{g}{g}\nput{90}{h}{h}\nput{90}{i}{i}\nput{90}{j}{j}\nput{90}{k}{k}
        \rput[c](-1,4){$\large{V_1}$}\rput[c](-1,1){$\large{V_2}$}
    \rput[c](0.2,4){$V_1'$}    \rput[c](0.2,0.2){$V_2'$}
  \end{pspicture}
  \fi
  \end{center}
  \caption{Visualization of the partition $V = V_1 \dot{\cup} \ldots \dot{\cup} V_k$ and the induced sets $V_{\ell}' \subseteq V_{\ell}$. Here $\prec$ is the transitive closure of the depicted digraph.}
  \end{figure}
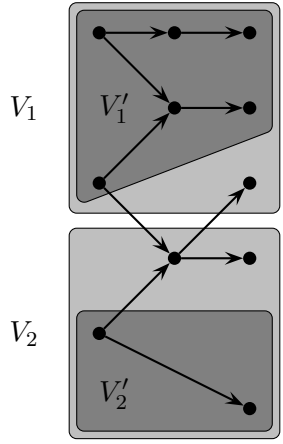
  
We now discuss the analysis. First we show that no cluster is more than a constant factor too large.
%(not quite bounded by $C$, but we can spend another constant factor here). 
\begin{lemma}
One has $|V_{\ell}'| \leq 2c$ for all $\ell=1,\ldots,k$.
\end{lemma}
\begin{proof}
  We know by Theorem~\ref{thm:ProbUSeperatedByClustering} that $\textrm{diam}(V_{\ell}') \leq \textrm{diam}(V_{\ell}) \leq \Delta < \frac{1}{2}$ 
  where the diameter is with respect to $d$.
  Fix a job $j^* \in V_{\ell}'$. 
  Then we know by Lemma~\ref{lem:SmallNeighborhoodOfClosePoints} that there are at most $2c$
  jobs $j$ with $d(j,j^*) \leq \frac{1}{2}$ and the claim follows.
\end{proof}

Next, we see that the clusters respect the precedence constraints.
\begin{lemma}
The solution $V_{1}',\ldots,V_{k}'$ is feasible in the sense that jobs on different machines do not have precedence constraints.
\end{lemma}
\begin{proof}
  Consider jobs processed on different machines, say (after reindexing) $j_1 \in V_{1}'$ and $j_2 \in V_2'$.
   If $j_1 \prec j_2$ then we did \emph{not} have $\Gamma^-(j_2) \subseteq V_2'$. This contradicts
  the definition of the sets $V_{\ell}'$.
\end{proof}

A crucial property that makes the algorithm work is that predecessors of some job $j \in J^*$ must be very close in $d$ distance.
\begin{lemma} \label{lem:DependentJobsInJStarAreCloseInMetricD}
For every $j_1,j_2 \in J^*$ with $j_1 \prec j_2$ one has $d(j_1,j_2) \leq \delta$.
\end{lemma}
\begin{proof}
  We know that
  \[
C^* \stackrel{j_1 \in J^*}{\leq} C_{j_1} \leq C_{j_1} + \underbrace{(1-y_{j_1,j_2})}_{=d(j_1,j_2)} \stackrel{LP}{\leq} C_{j_2} \stackrel{j_2 \in J^*}{\leq} C^* + \delta
  \]
  and so $d(j_1,j_2) \leq \delta$.
\end{proof}

We will use the three statements above together with Theorem~\ref{thm:ProbUSeperatedByClustering} to prove Lemma~\ref{lem:SchedulingOneIntervalViaCKR}.
%\begin{lemma}
%Every job $j^* \in J_0$ is contained in $V_1' \cup \ldots \cup V_k'$ with probability at least $1-8\ln(2n) \cdot \del%ta$.
%\end{lemma}
\begin{proof}[Proof of Lemma~\ref{lem:SchedulingOneIntervalViaCKR}]
  We have already proven that the scheduled blocks have size $|V_{\ell}'| \leq 2c$ and that there are no dependent
  jobs in different sets of $V_1',\ldots,V_k'$. To finish the analysis, 
we need to prove that a fixed job $j^* \in J^*$ is scheduled with good probability.
Consider the set $U := \{ j^* \} \cup (\Gamma^-(j^*) \cap J^*)$ of $j^*$ and its ancestors in $J^*$.

%Since the diameter of $U$ is at most $2\delta$ by Lemma~\ref{lem:DependentJobsInJStarAreCloseInMetricD},  for some fixed $u_1 \in U$  we have that $d(u_1, N(U,2 \Delta)) \leq 2\delta +  2\Delta$.
Since the diameter of $U$ is at most $2\delta$ by Lemma~\ref{lem:DependentJobsInJStarAreCloseInMetricD}, 
%  for every $j_1 \in U$  we have that $d(j^*, j_1) \leq 2\delta +  2\Delta$ by the triangle inequality.
we can use Lemma~\ref{lem:SmallNeighborhoodOfClosePoints}
to see that $ |N(U, \Delta/2)| \leq \frac{c}{1- 2\delta-  \Delta}$.
For our choice of $\Delta = 1/4$ and $\delta \leq \frac{1}{64 \log(4 c) }$, $|N(U,1/8)|\leq 2c$. 
From Theorem~\ref{thm:ProbUSeperatedByClustering},
the cluster is separated with probability at most 
 $\log(4c) \cdot  \frac{8 \delta}{\Delta} \leq \frac{1}{2}$.
\end{proof}

To schedule all jobs in $J^*$, we repeat the clustering procedure $O(\log m)$ times and simply schedule the remaining jobs on one machine. %, then show it is easy to schedule the remaining jobs.
%Recall that we let $m \leq n$ denote the minimum number of machines on which the instance can be scheduled while keeping the optimal makespan.

\begin{lemma} \label{lem: FullSchedulingOneIntervalViaCKR}
  Let $(x,y,C) \in Q(r)$ with $r \geq 5$. Let $C^* \geq 0$ and set $J^* \subseteq \{ j \in J \mid C^* \leq C_j < C^*+\delta\}$.
  Assume that all jobs in $\Gamma^{-}(J^*) \setminus J^*$ have been scheduled respecting precedence constraints.
  Then there is an algorithm with expected polynomial running time that schedules all jobs in $J^*$
  using at most $O(\log m) + \frac{|J^*|}{mc}$ many intervals. 
% Then, we can schedule all jobs in $J^*$ using at most $O(\log m) \cdot c + \frac{|J^*|}{m}$ time slots with probability at least $1-1/n^2$.
\end{lemma}

\begin{proof}
 Our algorithm in Lemma~\ref{lem:SchedulingOneIntervalViaCKR} schedules each $j \in J^*$ in an interval of length $2c$ with probability at least $1/2$.
 We run the algorithm for $2 \log m$ iterations, where input to iteration $k+1$ is the subset of jobs that are not scheduled in the first $k$ iterations. 
 For $k \in \{1,2,\ldots, 2 \log m \}$, let $J^{**}_k$ denote the subset of jobs that are scheduled in the $k^{th}$ iteration, and let $J^*_{k+1} := J^* \setminus \{ \bigcup^{k}_{k' = 1} J^{**}_{k'}$\}.
 In this notation,  $J^*_1 := J^*$.
  % and $J^{**}_0 := \emptyset$.
 Let  $\pazocal{S}(J^{**}_k)$ denote the schedule of jobs $J^{**}_k$  given by Lemma~\ref{lem:SchedulingOneIntervalViaCKR}.
 We schedule $\pazocal{S}(J^{**}_1)$ first, then for $k = 2,\ldots, 2 \log m$, we append the schedule $\pazocal{S}(J^{**}_k)$ after $\pazocal{S}(J^{**}_{k-1})$. %, while inserting $c$ new idle time slots in between.
Let $J' := J^{*}_{2\log m +1}$ denote the set of jobs that were not scheduled in the $2\log m$ iterations.
We schedule all jobs in $J'$ consecutively on a single machine after the completion of $\pazocal{S}(J^{**}_{2 \log m})$.

From our construction, the length of a schedule for  $J^*$, which is a random variable, is at most $O(\log m) +  \lceil \frac{|J'|}{c} \rceil$ many intervals.
For $k \in \{ 1,2, \ldots ,2\log m\}$,  Lemma~\ref{lem:SchedulingOneIntervalViaCKR}  guarantees that each job $j \in J^*_k$ gets scheduled in the $k^{th}$ iteration with probability at least $1/2$.
Therefore, the probability that $j \in J'$,  i.e., it does not get scheduled in the first $2 \log m$ iterations,  is at most $\frac{1}{2m}$.
This implies that $\mathbb{E}[|J'|] \leq \frac{|J^*|} {2m}$.
%We use the standard repetition argument to show that we can schedule all jobs in $J^*$ in at most $O(\log m) \cdot c + \frac{J^*}{m}$ time slots with probability at least $1-1/n^2$.
By Markov's inequality  $\text{Pr}[|J'| > \frac{|J^*|} {m} ] \leq \Pr[|J'| > 2 \cdot \mathbb{E}[|J'|] ] \leq 1/2$.
Hence we can repeat the described procedure until indeed we have a successful run with $|J'| \leq \frac{|J^*|}{m}$ which results in the claimed expected polynomial running time. 
%Consider repeating the above random experiment $2 \cdot \log n$ times.
%The probability that none of these experiments result in $|J'| <  \frac{|J^*|}{m}$ is at most $\frac{1}{n^2}$.
%Therefore, with probability at least $1-1/n^2$, in one of the experiments $|J'| <  \frac{|J^*|}{m}$. 
%Schedule the jobs as in this experiment.

Let us now argue that the schedule of $J^*$ is feasible. 
For $k \in \{ 1,2, \ldots, 2\log m\}$ and any two jobs $j, j' \in \pazocal{S}(J^{**}_k)$,  Lemma~\ref{lem:SchedulingOneIntervalViaCKR}  guarantees that precedence and communication constraints are satisfied.
Furthermore, Lemma~\ref{lem:SchedulingOneIntervalViaCKR} also ensures that there cannot be jobs  $j$, $j'$ such that  $j \in \pazocal{S}(J^{**}_k)$, $j' \in \pazocal{S}(J^{**}_{k'})$ and $j' \prec j$ and $k' > k$. Finally note that every length-$2c$ interval can be split into 2 length-$c$ intervals.
The claim follows.
%Finally, if there are two jobs $j$, $j'$ such that  $j \in \pazocal{S}(J^{**}_k)$, $j' \in \pazocal{S}(J^{**}_{k'})$ and $j \prec j'$  and $k' > k$, then the precedence and communication constraints are satisfied as we insert $c$ idle time slots between $\pazocal{S}(J^{**}_k)$ and $\pazocal{S}(J^{**}_{k'})$
% Thus, it only remains to  argue about the number of time slots used to schedule  $J^*$. 
\end{proof}

\subsection{The Complete Algorithm for $\p \infty \mid \Prec, p_j=1, \Intervals \mid C_{\max}$ }\label{sec:TheCompleteAlgorithm}

Now we have all the pieces to put the rounding algorithm together and prove its correctness.
We partition the jobs into batches, where each batch consists of subset of jobs that have  $C_j$ very close to each other in the LP solution. 
%For $\delta = \frac{1}{16 \log(2C)}$ define $k \in \{0\} \cup \mathbb{N}$, let $J_k = \{j \in J : k \cdot \delta \leq \theta_j \leq (k+1) \cdot \delta\}$.
The complete algorithm is given below.

\begin{center}
  \psframebox{
\begin{minipage}{14cm}
  \textsc{The Complete Algorithm} \vspace{1mm} \hrule \vspace{1mm}
  \begin{enumerate*}
  \item[(1)] Solve the LP and let $(x,y,C) \in Q(r)$ with $r \geq 5$.
  \item[(2)] For $\delta = \frac{1}{64 \log(4 c)}$ and $k \in \{0,1,2 \ldots \frac{S-1}{\delta}\}$, define \\
                 $J_k = \{j \in J : k \cdot \delta \leq C_j < (k+1) \cdot \delta\}$ 
  \item[(3)] FOR $k=0$ TO $\frac{S-1}{\delta}$ DO 
      \begin{enumerate*}
        \item[(4)] Schedule the jobs in $J_k$ using the algorithm in Subsection \ref{subsec:singlebatch}.
%        \item Insert $c$ new empty idle slots.
      \end{enumerate*}
 \end{enumerate*}
  \end{minipage}}
\end{center}

Now we finish the analysis of the rounding algorithm.
\begin{proof}[Proof of Theorem~\ref{thm:LPRoundingTheorem}]
  Let us quickly verify that the schedule constructed by our algorithm is feasible.
 For jobs $j_1 \prec j_2$ with $j_1 \in J_{k_1}$ and $j_2 \in J_{k_2}$, the LP implies that $C_{j_1} \leq C_{j_2}$
 and so $k_1 \leq k_2$. If $k_1<k_2$, then $j_1$ will be scheduled in an earlier interval than $j_2$.
 If $k_1=k_2=k$, then 
% For $k = 1,2, \ldots \frac{S-1}{\delta}$  and for any two jobs $j, j' \in J_k$,
 Lemma \ref{lem: FullSchedulingOneIntervalViaCKR}  guarantees that precedence  constraints are satisfied.
%Furthermore, 
%the LP constraints ensures that there cannot be jobs  $j$, $j'$ such that  $j \in J_k$, $j' \in J_{k'}$ and $j' \prec j$ and $k' > k$. 
%Finally, if there are two jobs $j$, $j'$ such that  $j \in J_k$, $j' \in J_{k'}$ and $j \prec j'$  and $k' > k$, then the precedence and communication constraints are satisfied in our schedule as we insert $c$ idle time slots between the schedule of jobs $J_k$ and $J_{k+1}$.

It remains to bound the makespan of our algorithm.  
Lemma \ref{lem: FullSchedulingOneIntervalViaCKR} guarantees that for $k \in \{0,1,2 \ldots \frac{S-1}{\delta}\}$, the jobs in $J_k$ are scheduled using at most $O(\log m) + \frac{|J_k|}{cm}$ many intervals. %time slots with probability at least $1-1/n^2$.
% By applying the union bound argument over all sets $J_k$, with probability at least $1-1/n$ the makespan of our schedule is at most
Then the total number of intervals required by the algorithm is bounded by
$$
\frac{S}{\delta} \cdot O(\log m)  + \sum^{\frac{S-1}{\delta}}_{k = 0} \frac{|J_k|}{cm} =O(\log m \cdot \log c)  \cdot S  +  \frac{|J|}{cm}\leq O(\log m \cdot \log c) \cdot S.
$$
Here we use that $|J| \leq S \cdot cm$ is implied by the constraints defining $K$.
%As both $(S-1) \cdot c$ and  $\frac{J}{m}$ are lowerbounds on the optimal schedule, we complete the proof that the makespan of our algorithm is at most $O(\log m \cdot \log c)$ times the optimal with probability $1-1/n$.
\end{proof}

\begin{remark}
We note that
it is possible to reverse-engineer our solution and write a more compact LP for the problem, enforcing only the necessary constraints such as those given by Lemmas~\ref{lem:metric} and~\ref{lem:clustercapacity}.
%Advantages of such an approach is that S-A hierarchy is needed only in the proof of validity of the LP relaxation, and the final LP has much fewer variables therefore can be solved faster.
Such an LP would be simpler and could be solved more efficiently.
However, we feel that the Sherali-Adams hierarchy gives a more principled and intuitive way to tackle the problem and explain how the LP arises, and hence we choose to present it that way.
\end{remark}

\section{Reductions}\label{sec:Reductions}

We now justify our earlier claim:
the special case $\p \infty \mid \Prec, p_j=1, \Intervals \mid C_{\max}$ indeed captures 
the full computational difficulty of the more general problem $\p \mid \Prec, c \mid C_{\max}$.
The main result for this section will be the following reduction:
\begin{theorem}
	\label{thm:ReductionFromGeneral}
  Suppose there is a polynomial time algorithm that takes a solution for the LP $Q(r)$ with parameters $m,c,S \in \setN$ and $r \geq 5$
  and transforms it into a schedule for $\p \infty \mid \Prec, p_j=1, \Intervals \mid C_{\max}$ using at most $\alpha \cdot S$
  intervals. Then there is a polynomial time  $O(\alpha)$-approximation for $\p \mid \Prec, c \mid C_{\max}$.
%	An $\alpha$-approximation for $\p \infty \mid \Prec, p_j=1, \Intervals \mid C_{\max}$ implies an $O(\alpha)$-approximation for $P \mid \Prec, c \mid C_{\max}$.
\end{theorem}
For the reduction we will make use of the very well known list scheduling algorithm by Graham~\cite{GrahamListScheduling1966}
that can be easily extended to the setting with communication delays.
Here the notation  $\sigma(j) = ([t,t+p_j),i)$ means that the job $j$ is
processed in the time interval $[t,t+p_j)$ on machine $i \in [m]$.
\begin{center}
  \psframebox{
\begin{minipage}{14cm}
  \textsc{Graham's List Scheduling} \vspace{1mm} \hrule \vspace{1mm}
  \begin{enumerate*}
  \item[(1)] Set $\sigma(j) := \emptyset$ for all $j \in J$
  \item[(2)] FOR $t=0$ TO $\infty$ DO FOR $i=1$ TO $m$ DO
    \begin{enumerate*}
    \item[(3)] Select any job $j \in J$ with $\sigma(j) = \emptyset$ where every $j' \prec j$ satisfies the following:
      \begin{itemize*}
        \item If $j'$ is scheduled on machine $i$ then $j'$ is finished at time $\leq t$
        \item If $j'$ is schedule on machine $i' \neq i$ then $j'$ finished at time $\leq t-c$
      \end{itemize*}
     \item[(4)] Set $\sigma(j) := ([t,t+p_j),i)$ (if there was such a job)
  \end{enumerate*}
  \end{enumerate*}
  \end{minipage}}
\end{center}
For example, for the problem $\p \mid \Prec \mid C_{\max}$, Graham's algorithm gives a 2-approximation.
The analysis works by proving that there is a chain of jobs covering all time units where not all machines are busy.
Graham's algorithm does \emph{not} give a constant factor approximation for our problem with communication delays,
but it will still be useful for our reduction.

% Our proof will involve considering chains in a fixed schedule. 
%Let $j_{\ell}$ be the job which finishes last in a schedule.
%Consider a \emph{chain} of jobs proceeding each other $j_1 \prec \cdots \prec j_{\ell}$, 
%starting at $j_1$, which has no predecessor, where $j_{i-1}$ is the predecessor of $j_i$ which finishes last.
%Given a schedule for jobs in $J$, we let $\pazocal{Q}(J)$ denote the set of chains.
%For any $Q \in \pazocal{Q}(J)$, the length of chain $Q$ is $|Q|$.
Recall that a set of jobs $\{j_1,\ldots,j_{\ell} \} \subseteq J$ with $j_{\ell} \prec j_{\ell -1} \prec \ldots \prec j_{1}$ is called a \emph{chain}.
We denote  $\pazocal{Q}(J)$ as the set of all chains in $J$ w.r.t. precedence order $\prec$.
\begin{lemma}
	\label{lem:GrahamListScheduling}
  Graham's list scheduling on an instance of $\p \mid \Prec, c \mid C_{\max}$
  results in a schedule with makespan at most
  $\frac{1}{m} \sum_{j\in J}p_j+\max_{Q \in \pazocal{Q}(J)}\{\sum_{j \in Q}p_j+c\cdot (|Q|-1) \}$.
\end{lemma}

  % Let $j_l$ be a job which finishes last and $t_{j_l}$ its starting time.
  % Let $j_{l-1}$ be a predecessor of $j_l$ which finishes last. Then, by our precedence constraint, we have $t_{j_l}\ge t_{j_{l-1}}+p_{j_{l-1}}$. Moreover, if they are scheduled on different machines, then $t_{j_l}\ge t_{j_{l-1}}+p_{j_{l-1}}+C$.
	\begin{figure}
		\begin{center}
		\ifrenderfigures
			\psset{unit=0.6cm}
			\begin{pspicture}(-4,-3)(18,4.5)
			\pnode(4,3.5){A}\pnode(6,3.5){B} \ncline[arrowsize=5pt]{|<->|}{A}{B} \naput{$C$}
			\psline[fillstyle=vlines](18,0)(0,0)(0,3)(18,3)
			\rput[c](0,0){\pspolygon[fillstyle=solid,fillcolor=lightgray](0,0)(2,0)(2,1)(0,1)}
			\rput[c](0,1){\pspolygon[fillstyle=solid,fillcolor=lightgray](0,0)(1,0)(1,1)(0,1)}
			\rput[c](1,1){\pspolygon[fillstyle=solid,fillcolor=lightgray](0,0)(1,0)(1,1)(0,1)}
			\rput[c](0,2){\pspolygon[fillstyle=solid,fillcolor=lightgray](0,0)(1,0)(1,1)(0,1)}
			\rput[c](1,2){\pspolygon[fillstyle=solid,fillcolor=lightgray](0,0)(1,0)(1,1)(0,1)}
			\rput[c](2,1){\pspolygon[fillstyle=solid,fillcolor=gray](0,0)(2,0)(2,1)(0,1)} \pnode(3,1.5){j1} \rput[c](j1){$j_{\ell}$}
			\rput[c](4,1){\pspolygon[fillstyle=solid,fillcolor=lightgray](0,0)(1,0)(1,1)(0,1)}
			\rput[c](5,1){\pspolygon[fillstyle=solid,fillcolor=lightgray](0,0)(1,0)(1,1)(0,1)}
			\rput[c](6,2){\pspolygon[fillstyle=solid,fillcolor=lightgray](0,0)(2,0)(2,1)(0,1)}
			\rput[c](6,1){\pspolygon[fillstyle=solid,fillcolor=lightgray](0,0)(1,0)(1,1)(0,1)}
			\rput[c](7,1){\pspolygon[fillstyle=solid,fillcolor=lightgray](0,0)(2,0)(2,1)(0,1)}
			\rput[c](6,0){\pspolygon[fillstyle=solid,fillcolor=lightgray](0,0)(3,0)(3,1)(0,1)}
			\rput[c](8,2){\pspolygon[fillstyle=solid,fillcolor=gray](0,0)(1,0)(1,1)(0,1)} \pnode(8.5,2.5){j2} \rput[c](j2){$\ldots$}
			\rput[c](10,2){\pspolygon[fillstyle=solid,fillcolor=gray](0,0)(2,0)(2,1)(0,1)} \pnode(11,2.5){j3} \rput[c](j3){$j_2$}
			\rput[c](9,2){\pspolygon[fillstyle=solid,fillcolor=lightgray](0,0)(1,0)(1,1)(0,1)}
			\rput[c](16,0){\pspolygon[fillstyle=solid,fillcolor=gray](0,0)(2,0)(2,1)(0,1)} \pnode(17,0.5){j4} \rput[c](j4){$j_{1}$}
			\rput[c](9,0){\pspolygon[fillstyle=solid,fillcolor=lightgray](0,0)(2,0)(2,1)(0,1)}
			\rput[c](9,0){\pspolygon[fillstyle=solid,fillcolor=lightgray](0,0)(2,0)(2,1)(0,1)}
			\rput[c](11,1){\pspolygon[fillstyle=solid,fillcolor=lightgray](0,0)(1,0)(1,1)(0,1)}
			\rput[c](11,0){\pspolygon[fillstyle=solid,fillcolor=lightgray](0,0)(2,0)(2,1)(0,1)}
			\rput[c](12,2){\pspolygon[fillstyle=solid,fillcolor=lightgray](0,0)(2,0)(2,1)(0,1)}
			\rput[c](14,2){\pspolygon[fillstyle=solid,fillcolor=lightgray](0,0)(3,0)(3,1)(0,1)}
			\rput[c](14,1){\pspolygon[fillstyle=solid,fillcolor=lightgray](0,0)(2,0)(2,1)(0,1)}
			\rput[c](14,0){\pspolygon[fillstyle=solid,fillcolor=lightgray](0,0)(2,0)(2,1)(0,1)}
			\ncline[nodesepA=10pt,nodesepB=10pt,arrowsize=5pt]{->}{j1}{j2}
			\ncline[nodesepA=10pt,nodesepB=10pt,arrowsize=5pt]{->}{j2}{j3}
			\ncline[nodesepA=10pt,nodesepB=10pt,arrowsize=5pt]{->}{j3}{j4}
			\rput[r](-1,-1){chain}
			\rput[c](0,-1){\psline[linestyle=dotted,linewidth=0.5pt](0,0)(18,0) \psline[linewidth=1.0pt](2,0)(4,0)\psline[linewidth=1.0pt](8,0)(9,0)\psline[linewidth=1.0pt](10,0)(12,0)\psline[linewidth=1.0pt](16,0)(18,0) }
			\rput[r](-1,-2){delay after chain}
			\rput[c](0,-2){\psline[linestyle=dotted,linewidth=0.5pt](0,0)(18,0) \psline[linewidth=1.0pt](4,0)(6,0)\psline[linewidth=1.0pt](9,0)(10,0)\psline[linewidth=1.0pt](12,0)(14,0) }
			\rput[r](-1,-3){guaranteed busy}
			\rput[c](0,-3){\psline[linestyle=dotted,linewidth=0.5pt](0,0)(18,0) \psline[linewidth=1.0pt](0,0)(2,0)\psline[linewidth=1.0pt](6,0)(8,0)\psline[linewidth=1.0pt](14,0)(16,0) }
			\rput[c](-0.5,2.5){$1$}  \rput[c](-0.5,1.5){$\vdots$} \rput[c](-0.5,0.5){$m$}
			\pnode(0,3){A}\pnode(0,4){B}\ncline[arrowsize=5pt]{->}{B}{A} \nput{90}{B}{time $0$}
			\pnode(18,3){A}\pnode(18,4){B}\ncline[arrowsize=5pt]{->}{B}{A} \nput{90}{B}{makespan}
			\end{pspicture}
			\fi
			\caption{Analysis of Graham's algorithm with communication delay $c$.} 
		\end{center}
  \end{figure}
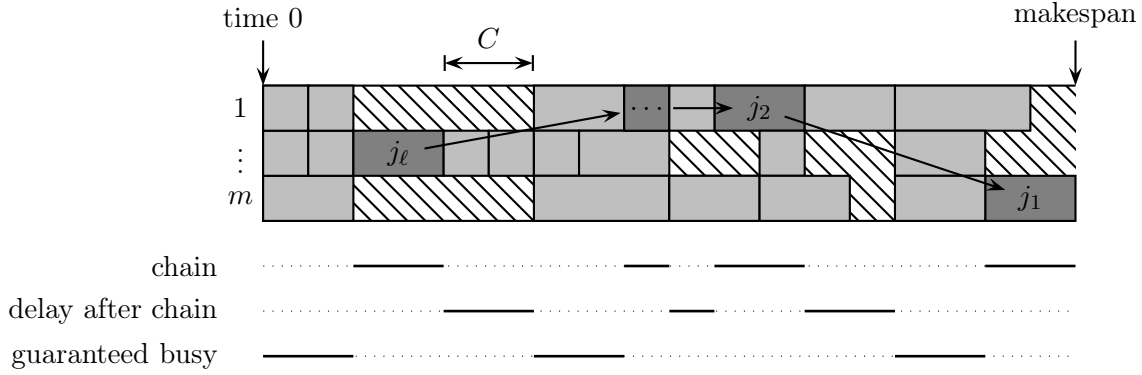
  \begin{proof} % $
We will show how to construct the chain $Q$ that makes the inequality hold. Let $j_1$ be the job which finishes last
in the schedule produced by Graham's algorithm and let $t_{j_1}$ be its start time. Let $j_2$ be the predecessor of $j_1$
that finishes last. More generally in step $i$, we denote $j_{i+1}$ as the predecessor of $j_i$ that finishes last.
The construction finishes with a job $j_{\ell}$ without predecessors. 
Now let $Q$ be the chain of jobs $j_{\ell} \prec j_{\ell-1} \prec \ldots \prec j_1$. 
The crucial observation is that for any $i \in \{ 1,\ldots,\ell-1\}$, 
either all machines are busy in the time interval $[t_{j_{i+1}}+p_{j_{i+1}}+c,t_{j_{i}})$  or this interval is empty. 
The reason is that Graham's algorithm does not leave unnecessary idle time and would have otherwise processed $j_{i}$ earlier.
It is also true that all $m$ machines are busy in the time interval $[0,t_{j_{\ell}})$.
%    Consider a chain $j_1 \prec \cdots \prec j_{\ell}$ in $\pazocal{Q}(J)$.
%  Let $t_{j_i}$ be the start time of job $j_i$ for $i \in [\ell]$.
%  As Graham's list scheduling leaves no unnecessary idle time on any machine, 
%  if $t_{j_i} \ge t_{j_{i-1}}+p_{j_{i-1}}+C$, then all machines are busy in time interval $[t_{j_{i-1}}+p_{j_{i-1}}+C,t_{j_i})$. 
%  Similarly, all machines are busy in $[0, t_{j_1})$. 
  The total amount of work processed in these busy time intervals is
	\[
	L := m\cdot\Big( t_{j_{\ell}}+ \sum_{i=1}^{\ell-1}\max \{t_{j_{i}}-(c+p_{j_{i+1}}+t_{j_{i+1}}), 0\}\Big)\le \sum_{j \in J} p_j - \sum_{i=1}^{\ell}p_{j_i}.
      \]
      Then any time between $0$ and the makespan falls into at least one of the following categories: 
      (a)~the busy time periods described above, (b)~the times that a job of the chain $Q$ is processed, 
      (c)~the interval of length $c$ following a job in the chain $Q$.
	Thus, we see that the makespan from Graham's list scheduling is at most
	% \[
	% m\cdot[t_{j_1}+ \sum_{i=1}^{l-1} (t_{j_{i+1}}-(C+p_{j_i}+t_{j_i}))]\le \sum_{j \in J} p_j - \sum_{i=1}^{l}p_{j_i},
	% \]
	\[
  t_{j_{1}}+p_{t_{j_{1}}} \leq \frac{L}{m} + \sum_{j \in Q} p_j + c \cdot (|Q|-1) \leq \frac{1}{m} \sum_{j \in J} p_j+\Big(1-\frac{1}{m}\Big)\sum_{j \in Q} p_{j}+c\cdot(|Q|-1).
	\]
	% and the claim follows.	
      \end{proof}

 It will also be helpful to note that the case of very small optimum makespan can be well approximated:      
\begin{lemma} \label{lem:PTASForProblemifOPTatMostC}
  Any instance for $\p \mid \Prec, c \mid C_{\max}$ with optimum objective function value at most $c$
  admits a PTAS.
\end{lemma}
\begin{proof}
  Let $J$ be the jobs in the instance and let $\textnormal{OPT}_m \leq c$ be the optimum value.
%  In this case, there is a way to schedule the jobs so that all dependent jobs are scheduled on the same machine.
  Consider the undirected graph $G = (J,E)$ with $\{ j_1,j_2 \} \in E \Leftrightarrow (( j_1 \prec j_2)\textrm{ or }(j_2 \prec j_1))$.
  Let $J = J_1 \dot{\cup} \ldots \dot{\cup} J_N$ be the partition of jobs into connected components w.r.t. graph $G$.
  We abbreviate $p(J') := \sum_{j \in J'} p_j$. The assumption guarantees that the optimum solution cannot afford to pay the communication delay
  and hence there is a length-$c$ schedule that assigns all jobs of the same connected component to the same machine. If we think of a connected component $J_{\ell}$
  as an ``item'' of size $p(J_{\ell})$, then for any fixed $\varepsilon>0$ we can use a PTAS for $\p \mid  \mid C_{\max}$ (i.e. makespan minimization without
  precedence constraints) to find a partition of ``items'' as $[N] = I_1 \dot{\cup} \ldots \dot{\cup} I_m$ with $\sum_{\ell \in I_i} p(J_{\ell}) \leq (1+\varepsilon) \cdot \textnormal{OPT}_m$ in polynomial time~\cite{ApproxAlgoForSchedulingHochbaumShmoysJACM87}. 
  Arranging the jobs $\bigcup_{\ell \in I_i} J_{\ell}$ on machine $i$ in any topological order finishes the argument.
\end{proof}

Additionally, it is a standard argument to convert an instance with arbitrary $p_j$ to an instance where all $p_j \leq n / \varepsilon$,
while only losing a factor of $(1+2\varepsilon)$ in the approximation. 
For $p_{\textrm{max}}:=\max_j p_j$, we simply scale the job lengths and communication delay down by a factor of $\frac{n }{ \varepsilon p_{\textrm{max}}}$
then round them to the nearest larger integer.
This results in at most a $2 \varepsilon$ fraction of the optimal makespan being rounded up and all 
job sizes are integral and at most $n / \varepsilon$. 

Now we can show the main reduction: 
\begin{proof}[Proof of Theorem~\ref{thm:ReductionFromGeneral}] 
	
  Consider an instance of $\p \mid \Prec, c \mid C_{\max}$ with $p_j, c \in \mathbb{N}$.
  Let $J$ denote its job set with precedence constraints, and $\textnormal{OPT}_m(J)$ denote its optimal value where $m$ is the number of
  available machines. 
  By the previous argument, we may assume that $p_j \leq 2n$ for all $j \in J$. 
  Moreover, by Lemma~\ref{lem:PTASForProblemifOPTatMostC} we only need to focus on the case where $\textnormal{OPT}_m(J) > c$.
  We may guess the optimum value of $\textnormal{OPT}_m(J)$ as $\textnormal{OPT}_m(J) \in \{ 1,\ldots,2n^2\}$. 
  %, since otherwise the problem can be solved in polynomial time. 
 % Define $J_S=\{j\in J\mid p_j\le c\}$ as the set of \emph{short} jobs and $J_L=\{j\in J\mid p_j> c\}$
 % as the set of \emph{long} jobs.

  Let $J'$ denote the job set obtained by splitting each job $j \in J$
  into a chain of $p_j$ unit sub-jobs $j^{(1)}\prec \cdots \prec j^{(p_j)}$. 
  Moreover, precedence constraints in $J$ are preserved in $J'$ as we set all 
  predecessors of $j$ to be predecessors of $j^{(1)}$ and all successors of $j$ to be successors of $j^{(p_j)}$, 
  see Figure~\ref{fig:SplittingJobsIntoUnitLengthJobs}.
  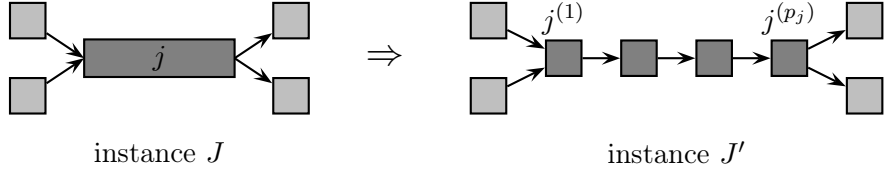
\begin{figure}
  \begin{center}
  \ifrenderfigures
    \psset{unit=0.5cm}
    \begin{pspicture}(0,-2)(8,2)
    \fnode[framesize=1,fillcolor=lightgray,fillstyle=solid](-1.5,1.5){a}
    \fnode[framesize=1,fillcolor=lightgray,fillstyle=solid](-1.5,-0.5){b}
    \fnode[framesize=1,fillcolor=lightgray,fillstyle=solid](5.5,1.5){c}
    \fnode[framesize=1,fillcolor=lightgray,fillstyle=solid](5.5,-0.5){d}
    \pspolygon[fillstyle=solid,fillcolor=gray](0,0)(4,0)(4,1)(0,1) \rput[c](2,0.5){$j$}
    \pnode(0,0.5){jstart}
    \pnode(4,0.5){jend}
    \ncline[arrowsize=5pt]{->}{a}{jstart}
    \ncline[arrowsize=5pt]{->}{b}{jstart}
   \ncline[arrowsize=5pt]{->}{jend}{c}
   \ncline[arrowsize=5pt]{->}{jend}{d}
   \rput[c](8,0.5){\Large $\Rightarrow$}
   \rput[c](2,-2){instance $J$}
     \end{pspicture}
     \begin{pspicture}(-4,-2)(8,2)
    \fnode[framesize=1,fillcolor=lightgray,fillstyle=solid](-1.5,1.5){a}
    \fnode[framesize=1,fillcolor=lightgray,fillstyle=solid](-1.5,-0.5){b}
    \fnode[framesize=1,fillcolor=gray,fillstyle=solid](0.5,0.5){j1} \nput[labelsep=2pt]{90}{j1}{$j^{(1)}$}
    \fnode[framesize=1,fillcolor=gray,fillstyle=solid](2.5,0.5){j2}
    \fnode[framesize=1,fillcolor=gray,fillstyle=solid](4.5,0.5){j3}
    \fnode[framesize=1,fillcolor=gray,fillstyle=solid](6.5,0.5){j4}\nput[labelsep=2pt]{90}{j4}{$j^{(p_j)}$}
    \fnode[framesize=1,fillcolor=lightgray,fillstyle=solid](8.5,1.5){c}
    \fnode[framesize=1,fillcolor=lightgray,fillstyle=solid](8.5,-0.5){d}
    \ncline[arrowsize=5pt]{->}{a}{j1}
    \ncline[arrowsize=5pt]{->}{b}{j1}
    \ncline[arrowsize=5pt]{->}{j1}{j2}
    \ncline[arrowsize=5pt]{->}{j2}{j3}
    \ncline[arrowsize=5pt]{->}{j3}{j4}
   \ncline[arrowsize=5pt]{->}{j4}{c}
   \ncline[arrowsize=5pt]{->}{j4}{d}
   \rput[c](3.5,-2){instance $J'$}
 \end{pspicture}
 \fi
 \caption{Splitting jobs into chains of unit-length jobs.\label{fig:SplittingJobsIntoUnitLengthJobs}}
  \end{center}
  \end{figure}
  We note that $\textnormal{OPT}_m(J') \leq \textnormal{OPT}_m(J)$ as splitting does not increase the value of the optimum.
  Let $\pazocal{S}_m(J')$ be a schedule achieving the value of $\textnormal{OPT}_m(J')$.
  Next, observe that $\pazocal{S}_m(J')$ implies an integral solution for $Q(r)$ with parameters $m,c,S$ where $S := \lceil \textnormal{OPT}_m(J)/c \rceil$ and $r := 5$.
  In particular here we use that if jobs $j_1 \prec j_2$ are scheduled on different machines by $\pazocal{S}_m(J')$, then their starting times differ by at least $c+1$ and hence they are assigned to different length-$c$ intervals.
  
%  Consider the instance of $\p\infty \mid \Prec, p_j=1, \Intervals \mid C_{\max}$ with job set $J'$. 
 % Note that we still count the makespan of a schedule and optimal value is the last time unit of the last time interval.
%  Let $\textnormal{OPT}_{\infty, \text{int}}(J')$ be the optimum value of this instance, where we count the last time unit of the last used interval. %denote its optimal value. 
%  Then $\textnormal{OPT}_{\infty,\text{int}}(J') \le \textnormal{OPT}_{m}(J)+c\le 2\textnormal{OPT}_{m}(J) $. 
  Now we execute the assumed $\alpha$-approximate rounding algorithm and obtain a schedule $\pazocal{S}_{\infty,\text{int}}(J')$
  that uses $T \leq \alpha S$ many intervals. 
%  with makespan 
%  $\textnormal{APX}_{\infty,\text{int}}(J')\le \alpha \cdot \textnormal{OPT}_{\infty,\text{int}}(J')$. 
%  Let $T$ denote the total number of intervals used in $\pazocal{S}_{\infty,\text{int}}(J')$, i.e. $\textnormal{APX}_{\infty,\text{int}}(J')= c\cdot T$. 
  We will use this solution  $\pazocal{S}_{\infty,\text{int}}(J')$ to construct
   a schedule $\pazocal{S}_{\infty}(J)$ for $\p\infty \mid \Prec,c \mid C_{\max}$ 
  with job set $J$ by running split sub-jobs consecutively on the same processor. 
  This will use $4T$ time intervals in total. Recall that $I_s$ denotes the time interval $[sc, (s+1)c)$.
  The rescheduling process is as follows:
	
  For a fixed job $j \in J$, let $I_{s_1}$ be the time interval where $j^{(1)}$ is scheduled in $\pazocal{S}_{\infty,\text{int}}(J')$. 
  Then all other sub-jobs of $j$ should be either scheduled in $I_{s_1}$ or the time intervals after $I_{s_1}$. 
\begin{itemize}	
	\item {\bf Case 1: Some sub-job of $j$ is not scheduled in $I_{s_1}$.}\\
  Schedule job $j$ at the beginning of time interval $I_{4{s_1}+1}$ on a new machine.
  If $j$ is a short job, then it will finish running by the end of the interval. 
  Otherwise $j$ is a long job.
  Let $I_{s_2}$ be the last time interval where a sub-job of $j$ is scheduled in $\pazocal{S}_{\infty,\text{int}}(J')$. 
  Then, the length satisfies $p_j \le c\cdot (s_2-s_1+1)$, 
  which implies that the job finishes by time $c\cdot (4s_1+1)+p_j\le c\cdot(4s_2-1)$.
	
  \item {\bf Case 2: All sub-jobs of $j$ are scheduled in $I_{s_1}$.}\\
  Simply schedule job $j$ during time interval $I_{4s_1}$ on the same machine as in $\pazocal{S}_{\infty,\text{int}}(J')$.  
  % Schedule job $j$ in $S'$ in time interval $I_{4s_1}$ on the same machine as scheduled in $S$. 
  % Note that there could be multiple jobs whose sub-jobs are all scheduled in the same time interval on the same machine, 
  % and in this case their total processing time is upper bounded by $C$.
  % So we can schedule them in arbitrary order on the same machine in one time interval in $S'$.
\end{itemize}	
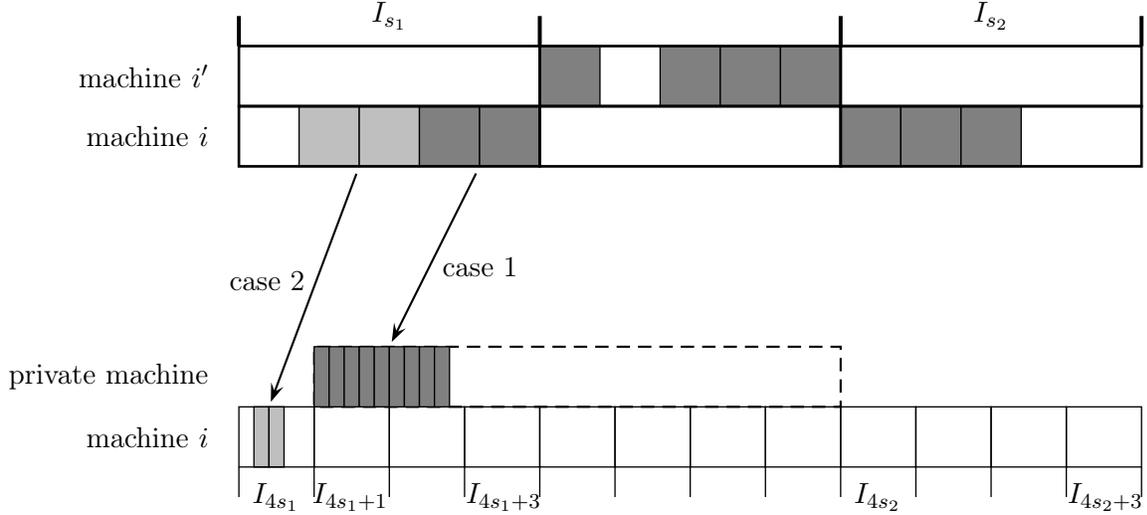
\begin{figure}
  \begin{center}
\ifrenderfigures
    \psset{unit=0.8cm}
    \begin{pspicture}(0,-2)(12,7)
      % UPPER PART
     \drawRect{fillstyle=solid,fillcolor=lightgray,linewidth=0.5pt}{1}{4}{1}{1}\drawRect{fillstyle=solid,fillcolor=lightgray,linewidth=0.5pt}{2}{4}{1}{1} \pnode(2,4){caseIIA}
    \drawRect{fillstyle=solid,fillcolor=gray,linewidth=0.5pt}{3}{4}{1}{1}\drawRect{fillstyle=solid,fillcolor=gray,linewidth=0.5pt}{4}{4}{1}{1}
    \drawRect{fillstyle=solid,fillcolor=gray,linewidth=0.5pt}{5}{5}{1}{1}\drawRect{fillstyle=solid,fillcolor=gray,linewidth=0.5pt}{7}{5}{1}{1}
    \drawRect{fillstyle=solid,fillcolor=gray,linewidth=0.5pt}{8}{5}{1}{1}\drawRect{fillstyle=solid,fillcolor=gray,linewidth=0.5pt}{9}{5}{1}{1}
    \drawRect{fillstyle=solid,fillcolor=gray,linewidth=0.5pt}{10}{4}{1}{1}\drawRect{fillstyle=solid,fillcolor=gray,linewidth=0.5pt}{11}{4}{1}{1}\drawRect{fillstyle=solid,fillcolor=gray,linewidth=0.5pt}{12}{4}{1}{1}
    \drawRect{linewidth=1.0pt}{0}{5}{5}{1}\drawRect{linewidth=1.0pt}{5}{5}{5}{1}\drawRect{linewidth=1.0pt}{10}{5}{5}{1} \rput[r](-0.5,5.5){machine $i'$}
    \drawRect{linewidth=1.0pt}{0}{4}{5}{1}\drawRect{linewidth=1.0pt}{5}{4}{5}{1}\drawRect{linewidth=1.0pt}{10}{4}{5}{1}  \rput[r](-0.5,4.5){machine $i$}
    \rput[c](2.5,6.5){$I_{s_1}$} \rput[c](12.5,6.5){$I_{s_2}$}
    \multido{\N=0+5}{4}{\psline[linewidth=1.5pt](\N,6)(\N,6.5)}
    % LOWER PART
    \drawRect{fillstyle=solid,fillcolor=lightgray,linewidth=0.5pt}{0.25}{-1}{0.25}{1} \pnode(0.5,0){caseIIB}
    \drawRect{fillstyle=solid,fillcolor=lightgray,linewidth=0.5pt}{0.5}{-1}{0.25}{1}
    \drawRect{linestyle=dashed}{1.25}{0}{8.75}{1}
 %   \multido{\N=0.00+1.25}{12}{\drawRect{linewidth=0.5pt}{\N}{0}{1.25}{1}}
    \multido{\N=0.00+1.25}{12}{\drawRect{linewidth=0.5pt}{\N}{-1}{1.25}{1}}
    \multido{\N=0.00+1.25}{13}{\psline[linewidth=0.5pt](\N,-1)(\N,-1.5)}
    \rput[r](-0.5,0.5){private machine}
    \rput[r](-0.5,-0.5){machine $i$}
    \rput[c](0.6125,-1.5){$I_{4s_1}$} \rput[c](1.85,-1.5){$I_{4s_1+1}$} \rput[c](4.4,-1.5){$I_{4s_1+3}$}
    \rput[c](10.6125,-1.5){$I_{4s_2}$} \rput[c](14.4,-1.5){$I_{4s_2+3}$}
    \ncline[nodesepA=3pt,nodesepB=2pt,arrowsize=5pt]{->}{caseIIA}{caseIIB} \nbput[labelsep=2pt]{case 2}
    \multido{\N=1.25+0.25}{9}{\drawRect{fillstyle=solid,fillcolor=gray,linewidth=0.5pt}{\N}{0}{0.25}{1}}
    \pnode(4,4){caseIA}\pnode(2.5,1){caseIB}
   \ncline[nodesepA=3pt,nodesepB=2pt,arrowsize=5pt]{->}{caseIA}{caseIB} \naput[labelsep=2pt]{case 1}
    \end{pspicture}   
\fi
  \caption{Transformation of the schedule $\pazocal{S}_{\infty,\text{int}}(J')$ (top) to $\pazocal{S}_{\infty}(J)$ (bottom), where  $\pazocal{S}_{\infty}(J)$ is compressed by a factor of 4. Here a ``private'' machine for a job $j$ means the machine never processes any job other than $j$.\label{fig:ScheduleTransformationStoSPrime}}
  \end{center}
\end{figure}

% It can be checked that the schedule $S'$ satisfies the precedence constraints and communication delays, which means that it 
 See Figure~\ref{fig:ScheduleTransformationStoSPrime} for a visualization.
 Then  $\pazocal{S}_{\infty}(J)$ is a valid schedule for $\p \infty \mid \Prec,c \mid C_{\max}$, with makespan $\le 4c\cdot T$.
	Moreover, $\pazocal{S}_{\infty}(J)$ satisfies the following:
	\begin{enumerate*}
		\item[(a)]A short job is fully contained in some interval $I_s$.
		
		\item[(b)]A long job's start time is at the beginning of some interval $I_s$.
	\end{enumerate*}
	
  For $\pazocal{S}_{\infty}(J)$, define a new job set $H$. 
  Every long job $j$ becomes an element of $H$ with its original running time $p_j$. 
  Meanwhile, every set of short jobs that are assigned to the same machine in one time interval becomes an element of $H$, 
  with running time equal to the sum of running times of the short jobs merged. 
  To summarize, a new job $h \in H$ corresponds to a set $h \subseteq J$ and $p_h = \sum_{j\in h} p_j$.
	
  We define the partial order $\tilde{\prec}$ on $H$ with $h_1 \tilde{\prec} h_2$ if and only if there are
  $j_1 \in h_1$ and $j_2 \in h_2$ with $j_1 \prec j_2$. 
%  some jobs comprising $j_1$ and $j_2$ in the original job set had precedent constraints. 
One can check that this partial order is well defined. 
  Moreover, by the fact that jobs assigned to the same interval but different machines do not have precedence constraints, 
  the length of the longest chain in $(H,\tilde{\prec})$ in terms of the number of elements is bounded by the number of intervals that are used, 
  which is at most $4T$.
	
  Now run Graham's list scheduling on the new job set $H$ with order $\tilde{\prec}$ and $m$ machines. 
  By Lemma~\ref{lem:GrahamListScheduling}, the makespan of the list scheduling is bounded by  
  $\frac{1}{m} \sum_{h\in H}p_h+\max_{Q \in \pazocal{Q}(H)}\{\sum_{h \in Q}p_h+ c\cdot |Q| \}$.
  As the total sum of the processing times does not change from $J$ to $H$, we see that
  $\frac{1}{m} \sum_{h\in H}p_h \leq \textnormal{OPT}_m(J)$. 
  Moreover, for any chain $Q \in \pazocal{Q}(H)$, $\sum_{h \in Q}p_h$ is no greater than the makespan of $\pazocal{S}_{\infty}(J)$, which is $4cT$.
	Finally, as argued earlier, the chain has $|Q| \leq 4T$ elements.
	Above all, 
  \begin{align*}
    \frac{1}{m} \sum_{h\in H}p_h+\max_{Q \in \pazocal{Q}(H)}\Big\{\sum_{h \in Q}p_h+c\cdot|Q| \Big\} &\le
          \textnormal{OPT}_m(J)+4cT+4cT \\
  &\le \textnormal{OPT}_m(J)+ 8\alpha \cdot cS \\
    &\le \textnormal{OPT}_m(J)+  16\alpha \cdot \textnormal{OPT}_{m}(J) \\
    &=O(\alpha)\cdot \textnormal{OPT}_{m}(J).
	\end{align*}
\end{proof}

\section{Minimizing Weighted Sum of Completion Times}

To illustrate the generality of our framework we show that it can be extended to handle different objective
functions, in particular we can minimize the \emph{weighted sum of completion times} of the jobs. Here we
restrict to the simplest case where jobs have unit length and an unbounded number of machines are available. 
In the 3-field notation, this problem is denoted by $\p \infty \mid \Prec, p_j=1, c \mid \sum_{j}w_j C_j$.
The input for this problem is the same as for the makespan minimization problem except that each job $j$ now has a weight $w_j \geq 0$. %\rem{T: $w_j \in \setZ_{\geq 0}$ wasn't needed.} 
%Without loss of generality, we assume that weights are positive integers.
The goal is to minimize the objective function $\sum_{j} w_j C_j$, where $C_j$ is the completion time of $j$, which is defined as the time slot in which job $j$ is scheduled. 
%Since our main aim is  to illustrate the generality of our framework, we only consider case when jobs have unit length and the number of machines is unbounded. 

Note that the LP $Q(r)$ has variables $C_j$ that denote the index of the length-$c$ interval where $j$ is being scheduled.
A natural approach would be to interpret $c \cdot C_j$ as the completion time of job $j$ and minimize $\sum_{j \in J} w_j \cdot c \cdot C_j$ over $Q(r)$. Then the rounding algorithm from Section~\ref{sec:approximation_for_pinfty}
will indeed schedule each job $j$ so that the completion time is at most $(O(\log c \cdot \log n) \cdot C_j+ \Theta(\log n)) \cdot c$.
We can observe that if a $O(\log c \cdot \log n)$ approximation is the goal, then this argument suffices for all jobs $j$
where the LP solution has $C_j \geq \Omega(\frac{1}{\log c})$ --- but it fails for jobs with $0 \leq C_j \ll 1$.

%Our algorithm for the completion time problem follows a similar framework as the makespan problem.
%In fact, it is easy to see that for all the jobs for which the makespan LP pays a cost of $\Theta(c)$ (that is, $C_j$ variables in the LP is at least $\Theta(1)$), the makespan algorithm already achieves $O(\log n \cdot \log c)$ approximation for the completion times of jobs.
%We only need to argue for those jobs whose completion time in the LP solution is significantly smaller than $c$. %\medskip

\subsection{The linear program}
In order to address this case, we first start with a more general  LP relaxation compared to the makespan result
which tracks the actual time slot where the jobs are processed, rather than just the interval.
Again, we use the parameter $m \in \setN$ to denote the number of machines that we allow the LP to use (one can set $m := n$)
and the parameter $S \in \setN$ to denote the number of intervals that we allow for the time horizon. We abbreviate $T := S \cdot c$ as the number of
time slots. Note that $T \leq nc$ always suffices for any non-idling schedule. We index time slots as $[T] := \{ 1,\ldots,T\}$
and consider an interval as a discrete set of slots $I_s := \{ cs+1,\ldots,c(s+1)\}$ where $s \in \{ 0,\ldots,S-1\}$.

Recall that in the makespan result $x_{j,i,s}$ variables indicated if job $j$ got scheduled on machine $i$ in the interval $s$. %\rem{T: In Sec 1-4 we use $x_{j,i,s}$ so I am switching to that order here as well.}
Here, we introduce additional variables of the form $z_{j,i,t}$ which indicate if job $j$ is scheduled on machine $i$ at time $t \in [T]$.
% (Note that completion time of any non-idling schedule is at most $nc$).
The variables $x_{j,i,s}$ are fully determined by summing over appropriate variables $z_{j,i,t}$, but we retain
them for  notational convenience.
Further, similar to our makespan result, we  impose an interval structure on the optimal solution and lose an $O(1)$ factor in the approximation ratio.
%Recall that $S$ denotes the number of intervals used in an optimal solution.

Let $\tilde{K}$ be the set of fractional solutions to the following LP. %\rem{T: This is not actually a relaxation. One needs to add an empty interval every $c$ time steps to make an integral solution for y $\p \infty \mid \Prec, p_j=1, c \mid \sum_{j}w_j C_j$ feasible.} %\rem{t: Should we give the LPs a different name like $\tilde{K}$ and $\tilde{Q}(r)$?}
\begin{eqnarray*}
  \sum_{i \in [m]} \sum_{t \in [T]} z_{j,i,t} &=& 1 \quad \forall j \in J \\
  \sum_{j \in J} z_{j,i,t} &\leq& 1 \quad \forall i \in [m] \;\; \forall t \in [T] \\
  \sum_{t' < t} \sum_{i \in [m]} z_{j_1,i,t'} &\geq& \sum_{t' \leq t} \sum_{i \in [m]} z_{j_2,i,t'} \quad \forall j_1 \prec j_2 \;\;\forall t \in [T] \\
   \sum_{t \in I_s} z_{j,i,t} &=& x_{j,i,s} \quad \forall j \in J \;\; \forall s \in \{ 0,\ldots,S-1\} \\
  0 \leq z_{j,i,t} &\leq& 1 \quad \forall j \in J, i \in [m], t \in [T] \\
\end{eqnarray*}

Similar to the makespan LP, let $\tilde{Q}(r)$ be the set of feasible solutions  $(x,y, z,C)$ to the following LP:
\begin{eqnarray*}
\text{Minimize}  &&\sum_{j \in J} w_j \cdot C_j \\
  y_{j_1,j_2} &=& \sum_{s \in \{ 0,\ldots,S-1\}} \sum_{i \in [m]} x_{(j_1,i,s),(j_2,i,s)} \\
  C_{j_2} &\geq& C_{j_1} + (1-y_{j_1,j_2}) \cdot c \quad \forall j_1 \prec j_2 \\
  C_j &=& \sum_{i \in [m]} \sum_{t \in [T]} z_{j,i,t} \cdot t \quad \forall j \in J \\
  (x,z) &\in& SA_r(\tilde{K})
\end{eqnarray*}

%As before, we require that $r \geq 5$. 
Note that the $C_j$ variables in this LP relaxation denote the actual completion time of $j$ unlike their role in the makespan result, where they were used to indicate the interval in which $j$ was scheduled.
% We also add the constraint that $C_j \geq \sum_{i} \sum_{t} z_{ijt} \cdot t$, although they may already be implied by our relaxation.
The main technical result for this section is the following:
\begin{theorem} \label{thm:WeightedSumOfComplTimeMainTechResult}
Consider an instance for $\p \infty \mid \Prec, p_j=1, c \mid \sum_{j}w_j C_j$ and 
a solution $(x,y,z,C) \in \tilde{Q}(r)$ with $r \geq 5$. % with parameters $m,c,S,T \in \setN$.
Then there is a randomized algorithm with expected polynomial running time that finds a feasible schedule so that
(i) $\E[C_j^A] \leq O(\log c \cdot \log n) \cdot C_j$ and (ii) $C_j^A \leq O(\log c \cdot \log n) \cdot C_j + O(\log n) \cdot c$ for all $j \in J$, where $C_j^A$ is the completion time of job $j$.
\end{theorem}
We briefly describe how Theorem~\ref{thm:WeightedSumOfComplTimeMainTechResult} implies the approximation algorithm
promised in Theorem~\ref{thm:maincomp}:
\begin{proof}[Proof of Theorem~\ref{thm:maincomp}]
Note that strictly speaking $\tilde{Q}(r)$ is not actually a relaxation of  $\p \infty \mid \Prec, p_j=1, c \mid \sum_{j}w_j C_j$. However one can take an optimum integral schedule and insert $c$ idle time slots every $c$ time units and obtain a feasible solution for $\tilde{Q}(r)$. This increases the completion time of any job by at most a factor of 2.
Then we set $r := 5$ and $m := n$ and solve the LP $\tilde{Q}(r)$ in time polynomial in $n$.
Now consider the randomized schedule from Theorem~\ref{thm:WeightedSumOfComplTimeMainTechResult} with completion times $C_j^A$. Then the expected objective function is
 $\E[\sum_{j \in J} w_j \cdot C^A_j] \leq O(\log n \cdot \log c) \cdot \big( \sum_{j \in J} w_j \cdot C_j \big)$.
 Markov's inequality guarantees that we can find in expected polynomial time a schedule that satisfies this inequality if we increase the right hand side by a constant factor. This completes the proof.
 %satisfies  $ \sum_{j} w_j \cdot C^A_j \leq O(\log n \cdot \log c) \cdot \left( \sum_{j} w_j \cdot C_j \right)$, which completes the proof.
\end{proof}

\subsection{The Rounding Algorithm}
Let $(x, y, z, C)$ be an optimal solution to the LP relaxation $\tilde{Q}(r)$ with $r \geq 5$.
It remains to show Theorem~\ref{thm:WeightedSumOfComplTimeMainTechResult}.
We partition the jobs based on their fractional completion times. 
For  $\delta = \frac{c}{64 \log(4c)}$
and  $k \geq 0$, let $J_k := \{j \in J : k \cdot \delta \leq C_j < k \cdot \delta\}$. %\rem{T: One could also define $J_k := \{ j \in J \mid \delta k \leq C_j \leq \delta (k+1)$ and $J_0 := \{ j \in J \mid C_j \leq \delta\}$. One can still schedule $J_1,J_2,\ldots$ with the existing algo and only treat $J_0$ special.}
%Define
%$$J_0 := \left \{j \in J:  C_j \leq c \right\}$$
%as the set of jobs with the fractional completion time in the LP solution not greater than $c$.
 
We give a separate  algorithm for scheduling jobs in $J_0$ within an interval of length at most $O(\log n) \cdot c$.
Now consider the remaining jobs.
For $k = 1,2,...$, we schedule jobs in the set $J_k$  using the algorithm from Section~\ref{sec:TheCompleteAlgorithm},
%<<<<<<< HEAD
%%\rem{JT: broken reference}
%=======
%>>>>>>> ba02ef7927f68512da1d66048ba00efb62217df1
inserting $c$ empty time slots between the schedule of jobs in the set $J_k$ and $J_{k+1}$.
Let $C^{A}_j$ denote the completion time of job $j$ in our algorithm.

\begin{lemma}
\label{lem:comp5}
For $k \geq 1$, consider a job $j \in J_k$. Then deterministically $C^{A}_j \leq O(\log n \cdot \log c) \cdot C_j $.
\end{lemma}

\begin{proof}
The claim follows from repeating the arguments in Lemma \ref{lem: FullSchedulingOneIntervalViaCKR}, so we only give a sketch here. 
Fix $k$ and consider scheduling the jobs in the set $J_k$ using the procedure described in Lemma \ref{lem: FullSchedulingOneIntervalViaCKR}, where we repeat the CKR clustering algorithm for $k = \{1,2, \ldots 2\log n\}$ iterations. 
Then the expected number of jobs that did not get scheduled in the first $2\log n$ iterations is at most $\frac{|J_k|}{n^2} < 1$. Therefore, in expected polynomial time we can find a schedule such that  $C^{A}_j  \in [2\log n \cdot O(c) \cdot k, 2\log n \cdot O(c) \cdot (k+1)]$.
From the definition of set $J_k$,  the fractional completion time $C_j$ of every job $j$ in $J_k$ is at least $k \cdot \frac{c}{64 \log (4c)}$ in the LP solution.
This completes the proof.
\end{proof}

The only new ingredient  for the completion time result is scheduling the jobs in the set $J_0$. 
For $j \in J_0$, let $t^*_j$ denote the earliest time instant $t$ at which the job is  scheduled to a fraction of at least $1-\varepsilon$ in the LP solution. Here $0<\varepsilon<1$ is a small constant that we determine later.
In scheduling theory this time is also called \emph{$\alpha$-point} with $\alpha = 1-\varepsilon$.
Formally 
\begin{equation} \label{eq:AlphaPoints}
  t^*_j := \min \left\{t' \in [T]:  \sum_{i=1}^m  \sum^{t'}_{t = 1} z_{j,i,t} \geq 1-\varepsilon   \right\}
\end{equation}

We use the same semimetric $d(j_1,j_2) := 1-y_{j_1,j_2}$ as in Section~\ref{sec:approximation_for_pinfty} and schedule jobs in $J_0$ using the following procedure.

\begin{center}
 \psframebox{
\begin{minipage}{14cm}
  \textsc{Schedule For $J_0$} \vspace{1mm} \hrule \vspace{1mm}
\begin{enumerate*}
  \item[(1)] Run a CKR clustering on the semimetric space $(J_0,d)$ with parameter $\Delta := \frac{1}{12}$ and let $V_1,\ldots,V_k$ be the clusters.
  \item[(2)] Let $V_{\ell}' := \{ j \in V_{\ell} \mid (\Gamma^-(j) \cap J_0) \subseteq V_{\ell}\}$ for $\ell = 1,\ldots,k$.
  \item[(3)] For all $\ell=1,\ldots,k$ assign jobs in $V_{\ell}'$ on  a single machine and schedule them in the increasing of order of $t^*_j$ values breaking ties in an arbitrary manner.
  \item[(4)] Insert a gap of $c$ time slots.
  \item[(5)] Let $J'_0 \subseteq J_0$ be the set of jobs that did not get scheduled in steps (1) - (3). Use Lemma \ref{lem: FullSchedulingOneIntervalViaCKR} to schedule $J'_0$.
  \end{enumerate*}
  \end{minipage}}
\end{center}
%\rem{T: I changed $\Delta$ to $\frac{1}{12}$.}
%The below lemma follows from Lemma \ref{lem:SchedulingOneIntervalViaCKR}.

\begin{lemma}\label{l:ctsecondbatch1}
  For a job $j_1 \in J_0$, the probability that $j_1$ gets scheduled in step (5) of the algorithm, 
i.e., $j_1 \in J'_0$, is at most  $O(\log c) \cdot \frac{C_{j_1}}{c}$.
\end{lemma}
\begin{proof}
  The arguments are a slight refinement of Lemma~\ref{lem:SchedulingOneIntervalViaCKR}.
  Consider the set $U := \{ j_1 \} \cup (\Gamma^{-}(j_1) \cap J_0)$ of $j_1$ and its ancestors. If $j_0 \prec j_1$, then $0 \leq C_{j_0} + c \cdot d(j_0,j_1) \leq C_{j_1}$ by the LP constraints and so $d(j_0,j_1) \leq \frac{C_{j_1}}{c}$.
  Then the diameter of $U$ with respect to semimetric $d$ is bounded by $2C_{j_1}$ and hence by Theorem~\ref{thm:ProbUSeperatedByClustering}.(b) the probability that $U$ is separated is bounded by $\ln(2|N(U,\Delta/2)|) \cdot \frac{4 \textrm{diam}(U)}{\Delta} \leq O(\log c) \cdot \frac{C_{j_1}}{c}$.
\end{proof}

The next lemma follows from repeating the arguments in Lemma \ref{lem:comp5}. %\rem{T: Isn't it EXPECTED completion time conditioned on $j \in J_0'$? S: Agreed? } \rem{JK: expected cost is $O(c)$ and as stated is also correct.}
\begin{lemma}
\label{l:ctsecondbatch2}
For a job $j \in J_0$, condition on the event that $j \in J'_0$.  Then, $C^A_j|_{(j \in J'_0)} \leq O(\log n ) \cdot c$.
\end{lemma}

%We assume that $(x,y,z,C)$ is a solution to $Q(r)$ with $r \geq 5$ and
We can now prove that every cluster $V_{\ell}'$ can be scheduled on one machine so that the completion time
of any job is at most twice the LP completion time.
%As before, $d$ is the semimetric defined by $d(j_1,j_2) := 1-y_{j_1,j_2}$. $
\begin{lemma} \label{l:ctsecondbatch}
  For a small enough constant $\varepsilon > 0$ ($\varepsilon = \frac{1}{12}$ suffices) the following holds: 
  Let $U \subseteq J$ be a set of jobs 
  with $\textrm{diam}(U) \leq \varepsilon$ w.r.t. distance $d$. Define $t_j^*$ as in Eq~\eqref{eq:AlphaPoints} and
%  \[
%    t^*_j := \min\Big\{ t' \in [T] \mid \sum_{i=1}^m \sum_{t=1}^{t'} z_{j,i,t} \geq 1-\varepsilon \Big\} 
%  \]
  schedule the jobs in $U$ in increasing order of $t_j^*$ on one machine and denote the completion time of $j$ by $C_j^A$.
  Then, in expectation $C_j^A \leq 2 t_j^*$ for every $j \in U$.
\end{lemma}

\begin{proof}
  Let us index the jobs in $U = \{ j_1,\ldots,j_{|U|}\}$ so that $t_{j_1}^* \leq \ldots \leq t_{j_{|U|}}^*$.
  Suppose for the sake of contradiction that there is some job $j_N$ with $C_{j_N}^A > 2 t_{j_N}^*$.
  Abbreviate $U^* := \{ j_1,\ldots,j_N\}$ and $\theta^* := t_{j_N}^*$ so that  $1 \leq t_j^* \leq \theta^*$ for $j \in U^*$.
  We observe that $|U^*| = \sum_{j \in U^*} p_j > 2\theta^*$. Then we have
  \[
(A) \;\; \sum_{j \in U^*} \sum_{i \in [m]} \sum_{t=1}^{\theta^*} z_{j,i,t} \geq (1-\varepsilon)|U^*|, \quad (B) \;\; \sum_{j \in U^*} y_{j,j_N} \geq (1-\varepsilon)|U^*|, \quad  (C) \;\; \sum_{i \in [m]} \sum_{t=1}^{\theta^*} z_{j_N,i,t} \geq 1-\varepsilon
\]
where $(A)$ and $(C)$ are by definition of $t_j^*$ and $(B)$ follows from $\textrm{diam}(U^*) \leq \textrm{diam}(U) \leq \varepsilon$.
Intuitively, this means that we have $|U^*| > 2\theta^*$ many jobs that the LP schedules almost fully
on slots $\{1,\ldots,\theta^*\}$ while (B) means that the jobs are almost fully scheduled on the same machine.
%<<<<<<< HEAD
%As before, we will use the properties of the Sherali-Adams hiearchy to formally derive a contradiction.
%\rem{JT: by what Lemma?}
%We know by Lemma~ that there is a distribution  $(\tilde{x},\tilde{z},\tilde{y}) \sim \pazocal{D}(j_N)$ such that $\E[\tilde{x}_{j,i,s}] = x_{j,i,s}$, $\E[\tilde{z}_{j,i,t}]=z_{j,i,t}$ and $\E[\tilde{y}_{j_1,j_2}]=y_{j_1,j_2}$ while the variables involving job $j_N$ are integral, i.e.
%=======
As before, we will use the properties of the Sherali-Adams hierarchy to formally derive a contradiction.
We know by Lemma~\ref{lem:PropertiesOfSAforSchedulingWithCommDelaysLP}\footnote{Strictly speaking, Lemma~\ref{lem:PropertiesOfSAforSchedulingWithCommDelaysLP} describes the SA properties for LP $Q(r)$, but an absolutely analogous statement holds for $\tilde{Q}(r)$.} that there is a distribution  $(\tilde{x},\tilde{z},\tilde{y}) \sim \pazocal{D}(j_N)$ so that $\E[\tilde{x}_{j,i,s}] = x_{j,i,s}$, $\E[\tilde{z}_{j,i,t}]=z_{j,i,t}$ and $\E[\tilde{y}_{j_1,j_2}]=y_{j_1,j_2}$ while the variables involving job $j_N$ are integral, i.e.
%>>>>>>> Section 5 and section 3 typos
$\tilde{x}_{j_N,i,s},\tilde{z}_{j_N,i,t} \in \{ 0,1\}$.
Consider the three events
  \[
(A') \sum_{j \in U^*} \sum_{i \in [m]} \sum_{t=1}^{\theta^*} \tilde{z}_{j,i,t} \geq (1-3\varepsilon)|U^*|, \quad (B')  \sum_{j \in U^*} \tilde{y}_{j,j_N} \geq (1-3\varepsilon)|U^*|, \quad   (C') \sum_{i \in [m]} \sum_{t=1}^{\theta^*} \tilde{z}_{j_N,i,t} = 1.
\]
Then by Markov inequality $\Pr[A'] \geq \frac{2}{3}$, $\Pr[B'] \geq \frac{2}{3}$ and $\Pr[C'] \geq 1-\varepsilon$, and so by the union bound
$\Pr[A' \wedge B' \wedge C'] > 0$, assuming $\varepsilon < \frac{1}{3}$.
Fix an outcome for $(\tilde{x},\tilde{y},\tilde{z})$ where the events $A',B',C'$ happen and let $i_N \in [m],t_N \in \{ 1,\ldots,\theta^*\}$ be the indices with $\tilde{z}_{j_N,i_N,t_N}=1$. Then the interval index with $t_N \in I_{s_N}$ satisfies $\tilde{x}_{j_N,i_N,s_N}=1$. Hence
\begin{eqnarray*}
  (1-3\varepsilon) |U^*| &\stackrel{(B')}{\leq}& \sum_{j \in U^*} \tilde{y}_{j,j_N} \stackrel{LP}{=} \sum_{j \in U^*} \sum_{i \in [m]} \sum_{s \in \{ 0,\ldots,S-1\}} \tilde{x}_{(j,i,s),(j_N,i,s)} \stackrel{\tilde{x}_{j_N,i_N,s_N}=1}{=} \sum_{j \in U^*} \tilde{x}_{j,i_N,s_N} \\
  &\stackrel{LP}{\leq}& \underbrace{\sum_{j \in U^*} \sum_{t=1}^{\theta^*} \tilde{z}_{j,i_N,t}}_{\leq \theta^*\textrm{ by LP}} + \underbrace{\sum_{j \in U^*} \sum_{t > \theta^*} \tilde{z}_{j,i_N,t}}_{\leq 3\varepsilon |U^*|\textrm{ by }(A')} \leq \theta^* + 3\varepsilon |U^*|
\end{eqnarray*}
Rearranging gives $|U^*| \leq \frac{1}{1-6\varepsilon} \theta^*$, which is a contradiction for $\varepsilon \leq \frac{1}{12}$.
\end{proof}

\begin{lemma}
\label{l:jnot}
For a job $j \in J_0$, condition on the event that it got scheduled in the step (3) of the algorithm. 
Then, $C^A_j |_{(j \not \in J'_0)} \leq O(C_{j})$.
\end{lemma}
\begin{proof}
  Consider a job $j \in J_0 \setminus J_0'$. Then $j \in V_{\ell}'$ and by construction, the
  set $V_{\ell}'$ has diameter at most $\Delta = \frac{1}{12}$ w.r.t. $d$.
  Then Lemma~\ref{l:ctsecondbatch} guarantees that the completion time is $C_j^A \leq 2t_j^*$ where
  we set $\varepsilon = \frac{1}{12}$.
  Finally note that an $\varepsilon$-fraction of $j$ was finished at time $t_j^*$ or later and hence
  $C_j = \sum_{i \in [m]} \sum_{t \in [T]} z_{j,i,t} \cdot t \geq \varepsilon \cdot t_j^*$. 
 Putting everything together we obtain $C_j^A \leq \frac{2}{\varepsilon} C_j$.
\end{proof}

\begin{comment}% T: New proof.
\begin{proof}
Fix a job $j_1 \in J_0$ and suppose it got scheduled in the step (3) of the algorithm.
Then, $C^A_j |_{(j \not \in J'_0)} \leq |U(j_1)| $. Note that diameter of every cluster in our algorithm is at most $\Delta$, hence for every $j \in U(j_1)$ we have $d(j_1, j) \leq \Delta$. This implies that $U(j_1) \subseteq A(j_1)$.
Proof now follows from the previous lemma.
\end{proof}
\end{comment}

We have everything to finish the proof of the completion time result.

\begin{proof}[ Proof of Theorem \ref{thm:WeightedSumOfComplTimeMainTechResult}]
From Lemma \ref{lem:comp5}, for $k \geq 1$ and  $ j \in J_k$,  we have deterministically $C^{A}_j \leq O(\log n \cdot \log c) \cdot C_j$.
Now consider a job $j \in J_0$. Then,
\begin{eqnarray*}
\E[C^A_j]  &=& \E[C^A_j | (j \not \in J'_0)] \cdot \text{Pr}[ (j \not \in J'_0)] + \E[C^A_j | (j  \in J'_0)] \cdot \text{Pr}[ (j \in J'_0)]  \\
&\leq& O(C_j) + O(\log c) \cdot \frac{C_j}{c}  \cdot O(\log n) \cdot c \quad (\text{from Lemmas } \ref{l:ctsecondbatch1}, \ref{l:ctsecondbatch2}, \ref{l:jnot} ) \\
&\leq& O(\log n \cdot \log c) \cdot O(C_j)
\end{eqnarray*}
Finally note that the completion time of a job $j \in J_0$ is always bounded by $C_j^A \leq O(\log n) \cdot c$. The claim follows.
%where in the last inequality we used that $\kappa_j\leq \frac{1}{12}$ and that the LP constraints guarantee $C_j \geq \kappa_j \cdot c$. 
%Therefore, $ \sum_{j} w_j \cdot \E[C^A_j] \leq O(\log n \cdot \log c) \cdot \left( \sum_{j} w_j \cdot C_j \right)$.
%Markov's inequality guarantees that we can find in expected polynomial time a schedule that satisfies  $ \sum_{j} w_j \cdot C^A_j \leq O(\log n \cdot \log c) \cdot \left( \sum_{j} w_j \cdot C_j \right)$, which completes the proof.
\end{proof}

\section*{Discussion and Open Problems}
We gave a new framework for scheduling jobs with precedence constraints and communication delays based on metric space clustering.
Our results take the first step towards resolving several important problems in this area. 
One immediate open question is to understand whether our approach can yield a constant-factor approximation for $\p \mid \Prec, c \mid \Cmax$.
A more challenging problem is to handle {\em non-uniform} communication delays in 
the problem $\p \mid \Prec, c_{jk} \mid \Cmax$, where $c_{jk}$ is the communication delay between jobs $j \prec k$.

\bibliographystyle{alpha}
\bibliography{prec-scheduling}

\newcommand{\etalchar}[1]{$^{#1}$}
\begin{thebibliography}{GKMP08}

\bibitem[AMMS08]{ambuhl2008precedence}
Christoph Amb{\"u}hl, Monaldo Mastrolilli, Nikolaus Mutsanas, and Ola Svensson.
\newblock Precedence constraint scheduling and connections to dimension theory
  of partial orders.
\newblock In {\em Bulletin of the European Association for Theoretical Computer
  Science (EATCS}. Citeseer, 2008.

\bibitem[Ban17]{Bansalmapsp}
N.~Bansal.
\newblock Scheduling open problems: Old and new.
\newblock MAPSP 2017. http://www.mapsp2017.ma.tum.de/MAPSP2017-Bansal.pdf,
  2017.

\bibitem[BGK96]{BampisGK96}
E.~Bampis, A.~Giannakos, and J.C. K{\"o}nig.
\newblock On the complexity of scheduling with large communication delays.
\newblock {\em European Journal of Operational Research}, 94(2):252 -- 260,
  1996.

\bibitem[BK10]{BansalK10}
N.~Bansal and S.~Khot.
\newblock Inapproximability of hypergraph vertex cover and applications to
  scheduling problems.
\newblock In {\em Automata, Languages and Programming, 37th International
  Colloquium, {ICALP} 2010, Bordeaux, France, July 6-10, 2010, Proceedings,
  Part {I}}, volume 6198 of {\em Lecture Notes in Computer Science}, pages
  250--261. Springer, 2010.

\bibitem[CC91]{ColinC91}
J.~Y. Colin and P.~Chr\'etienne.
\newblock C.p.m. scheduling with small communication delays and task
  duplication.
\newblock {\em Operations Research}, 39(4):680--684, 1991.

\bibitem[CKR04]{DBLP:journals/siamcomp/CalinescuKR04}
G.~C{\u{a}}linescu, H.~Karloff, and Y.~Rabani.
\newblock Approximation algorithms for the 0-extension problem.
\newblock {\em {SIAM} J. Comput.}, 34(2):358--372, 2004.

\bibitem[CZM{\etalchar{+}}11]{Chowdhury}
M.~Chowdhury, M.~Zaharia, J.~Ma, M.~I. Jordan, and I.~Stoica.
\newblock Managing data transfers in computer clusters with orchestra.
\newblock In {\em Proceedings of the ACM SIGCOMM 2011 Conference}, SIGCOMM ?11,
  page 98?109. Association for Computing Machinery, 2011.

\bibitem[Dro09]{Drozdowski09}
M.~Drozdowski.
\newblock {\em Scheduling with Communication Delays}, pages 209--299.
\newblock Springer London, London, 2009.

\bibitem[FKK{\etalchar{+}}14]{LP-for-DST-FriggstadKKLST-IPCO14}
Z.~Friggstad, J.~K{\"{o}}nemann, Y.~Kun{-}Ko, A.~Louis, M.~Shadravan, and
  M.~Tulsiani.
\newblock Linear programming hierarchies suffice for directed steiner tree.
\newblock In {\em Integer Programming and Combinatorial Optimization - 17th
  International Conference, {IPCO} 2014, Bonn, Germany, June 23-25, 2014.
  Proceedings}, pages 285--296, 2014.

\bibitem[FRT04]{TreeMetricFakcharoenpholRaoTalwar-JCSS04}
J.~Fakcharoenphol, S.~Rao, and K.~Talwar.
\newblock A tight bound on approximating arbitrary metrics by tree metrics.
\newblock {\em J. Comput. Syst. Sci.}, 69(3):485--497, 2004.

\bibitem[GFC{\etalchar{+}}12]{guo2012spotting}
Zhenyu Guo, Xuepeng Fan, Rishan Chen, Jiaxing Zhang, Hucheng Zhou, Sean
  McDirmid, Chang Liu, Wei Lin, Jingren Zhou, and Lidong Zhou.
\newblock Spotting code optimizations in data-parallel pipelines through
  periscope.
\newblock In {\em Presented as part of the 10th $\{$USENIX$\}$ Symposium on
  Operating Systems Design and Implementation ($\{$OSDI$\}$ 12)}, pages
  121--133, 2012.

\bibitem[GK07]{GiroudeauKoenig07}
R.~Giroudeau and J.C. K{\"o}nig.
\newblock Scheduling with communication delays.
\newblock In Eugene Levner, editor, {\em Multiprocessor Scheduling}, chapter~4.
  IntechOpen, Rijeka, 2007.

\bibitem[GKMP08]{GiroudeauKMP08}
R.~Giroudeau, J.C. K{\"o}nig, F.~K. Moulai, and J.~Palaysi.
\newblock Complexity and approximation for precedence constrained scheduling
  problems with large communication delays.
\newblock {\em Theoretical Computer Science}, 401(1):107 -- 119, 2008.

\bibitem[GLLK79]{GLLR79}
R.~L. Graham, E.~L. Lawler, J.~K. Lenstra, and A.~H. G.~Rinnooy Kan.
\newblock {Optimization and approximation in deterministic sequencing and
  scheduling: a survey}.
\newblock {\em Ann. Discrete Math.}, 4:287--326, 1979.

\bibitem[Gra66]{GrahamListScheduling1966}
R.~L. Graham.
\newblock Bounds for certain multiprocessing anomalies.
\newblock {\em Bell System Technical Journal}, 45(9):1563--1581, 1966.

\bibitem[HCB{\etalchar{+}}19]{huang2019gpipe}
Y.~Huang, Y.~Cheng, A.~Bapna, O.~Firat, D.~Chen, M.~Chen, H.~Lee, J.~Ngiam,
  Q.~V. Le, Y.~Wu, and Z.~Chen.
\newblock Gpipe: Efficient training of giant neural networks using pipeline
  parallelism.
\newblock In {\em Advances in Neural Information Processing Systems}, 2019.

\bibitem[HCG12]{hong2012finishing}
Chi-Yao Hong, Matthew Caesar, and P~Brighten Godfrey.
\newblock Finishing flows quickly with preemptive scheduling.
\newblock {\em ACM SIGCOMM Computer Communication Review}, 42(4):127--138,
  2012.

\bibitem[HLV94]{HoogeveenLV94}
J.A. Hoogeveen, J.K. Lenstra, and B.~Veltman.
\newblock Three, four, five, six, or the complexity of scheduling with
  communication delays.
\newblock {\em Operations Research Letters}, 16(3):129--137, 1994.

\bibitem[HM01]{HanenMunier73Apx}
C.~Hanen and A.~Munier.
\newblock An approximation algorithm for scheduling dependent tasks on m
  processors with small communication delays.
\newblock {\em Discrete Applied Mathematics}, 108(3):239 -- 257, 2001.

\bibitem[HS87]{ApproxAlgoForSchedulingHochbaumShmoysJACM87}
Dorit~S. Hochbaum and David~B. Shmoys.
\newblock Using dual approximation algorithms for scheduling problems
  theoretical and practical results.
\newblock {\em J. {ACM}}, 34(1):144--162, 1987.

\bibitem[JKS93]{JungKS93}
H.~Jung, L.M. Kirousis, and P.~Spirakis.
\newblock Lower bounds and efficient algorithms for multiprocessor scheduling
  of directed acyclic graphs with communication delays.
\newblock {\em Information and Computation}, 105(1):94 -- 104, 1993.

\bibitem[KLTY20]{KulkarniLTY20}
J.~Kulkarni, S.~Li, J.~Tarnawski, and M.~Ye.
\newblock {\em Hierarchy-Based Algorithms for Minimizing Makespan under
  Precedence and Communication Constraints}, pages 2770--2789.
\newblock 2020.

\bibitem[KMN11]{Hierarchies-for-Knapsack-KarlinMathieuNguyen-IPCO11}
A.~Karlin, C.~Mathieu, and C.~Thach Nguyen.
\newblock Integrality gaps of linear and semi-definite programming relaxations
  for knapsack.
\newblock In {\em Integer Programming and Combinatoral Optimization - 15th
  International Conference, {IPCO} 2011, New York, NY, USA, June 15-17, 2011.
  Proceedings}, pages 301--314, 2011.

\bibitem[Lau03]{Comparison-of-Hierarchies-Laurent-MOR03}
M.~Laurent.
\newblock A comparison of the sherali-adams, lov{\'{a}}sz-schrijver, and
  lasserre relaxations for 0-1 programming.
\newblock {\em Math. Oper. Res.}, 28(3):470--496, 2003.

\bibitem[LLKS93]{lawler1993sequencing}
Eugene~L Lawler, Jan~Karel Lenstra, Alexander HG~Rinnooy Kan, and David~B
  Shmoys.
\newblock Sequencing and scheduling: Algorithms and complexity.
\newblock {\em Handbooks in operations research and management science},
  4:445--522, 1993.

\bibitem[LR02]{LepereR02}
R.~Lepere and C.~Rapine.
\newblock An asymptotic o(ln $\rho$/ ln ln $\rho$)-approximation algorithm for
  the scheduling problem with duplication on large communication delay graphs.
\newblock In Helmut Alt and Afonso Ferreira, editors, {\em STACS 2002}, pages
  154--165, Berlin, Heidelberg, 2002. Springer Berlin Heidelberg.

\bibitem[LR16]{LeveyR16}
E.~Levey and T.~Rothvoss.
\newblock A (1+epsilon)-approximation for makespan scheduling with precedence
  constraints using {LP} hierarchies.
\newblock In {\em Proceedings of the 48th Annual {ACM} {SIGACT} Symposium on
  Theory of Computing, {STOC} 2016, Cambridge, MA, USA, June 18-21, 2016},
  pages 168--177, 2016.

\bibitem[LRK78]{LR78}
J.~K. Lenstra and A.~H.~G. Rinnooy~Kan.
\newblock Complexity of scheduling under precedence constraints.
\newblock {\em Oper. Res.}, 26(1):22--35, February 1978.

\bibitem[LYZ{\etalchar{+}}16]{luo2016towards}
Shouxi Luo, Hongfang Yu, Yangming Zhao, Sheng Wang, Shui Yu, and Lemin Li.
\newblock Towards practical and near-optimal coflow scheduling for data center
  networks.
\newblock {\em IEEE Transactions on Parallel and Distributed Systems},
  27(11):3366--3380, 2016.

\bibitem[MH97]{HanenMunierDuplication}
A.~Munier and C.~Hanen.
\newblock Using duplication for scheduling unitary tasks on m processors with
  unit communication delays.
\newblock {\em Theoretical Computer Science}, 178(1):119 -- 127, 1997.

\bibitem[Mic18]{michael2018scheduling}
L~Pinedo Michael.
\newblock {\em Scheduling: theory, algorithms, and systems}.
\newblock Springer, 2018.

\bibitem[MK97]{MunierKonig}
A.~Munier and J.C. K\"onig.
\newblock A heuristic for a scheduling problem with communication delays.
\newblock {\em Operations Research}, 45(1):145--147, 1997.

\bibitem[MN06]{RamseyPartitions-MendelNaorFOCS06}
M.~Mendel and A.~Naor.
\newblock Ramsey partitions and proximity data structures.
\newblock In {\em 47th Annual {IEEE} Symposium on Foundations of Computer
  Science {(FOCS} 2006), 21-24 October 2006, Berkeley, California, USA,
  Proceedings}, pages 109--118, 2006.

\bibitem[NHP{\etalchar{+}}19]{narayanan2018pipedream}
D.~Narayanan, A.~Harlap, A.~Phanishayee, V.~Seshadri, N.~Devanur, G.~Ganger,
  P.~Gibbons, and M.~Zaharia.
\newblock Pipedream: Generalized pipeline parallelism for dnn training.
\newblock In {\em {Proc. 27th ACM Symposium on Operating Systems Principles
  (SOSP)}}, {Huntsville, ON, Canada}, October 2019.

\bibitem[PST04]{PruhsST04}
Kirk Pruhs, Jir{\'{\i}} Sgall, and Eric Torng.
\newblock Online scheduling.
\newblock In Joseph~Y.{-}T. Leung, editor, {\em Handbook of Scheduling -
  Algorithms, Models, and Performance Analysis}. Chapman and Hall/CRC, 2004.

\bibitem[PY90]{PapadimitriouY90}
C.~H. Papadimitriou and M.~Yannakakis.
\newblock Towards an architecture-independent analysis of parallel algorithms.
\newblock {\em SIAM J. Comput.}, 19(2):322–328, April 1990.

\bibitem[Rot11]{DirectedSteinerTreeAndLasserre-RothvossArxiv2011}
T.~Rothvo{\ss}.
\newblock Directed steiner tree and the lasserre hierarchy.
\newblock {\em CoRR}, abs/1111.5473, 2011.

\bibitem[RS87]{RAYWARDSMITH1987}
V.J. Rayward-Smith.
\newblock Uet scheduling with unit interprocessor communication delays.
\newblock {\em Discrete Applied Mathematics}, 18(1):55 -- 71, 1987.

\bibitem[Sve09]{svensson2009approximability}
Ola Nils~Anders Svensson.
\newblock {\em Approximability of some classical graph and scheduling
  problems}.
\newblock PhD thesis, Universit{\`a} della Svizzera italiana, 2009.

\bibitem[Sve10]{Svensson10}
O.~Svensson.
\newblock Conditional hardness of precedence constrained scheduling on
  identical machines.
\newblock In {\em Proceedings of the Forty-Second ACM Symposium on Theory of
  Computing}, STOC '10, pages 745--754, New York, NY, USA, 2010. Association
  for Computing Machinery.

\bibitem[SW99]{SW99a}
P.~Schuurman and G.~J. Woeginger.
\newblock Polynomial time approximation algorithms for machine scheduling: Ten
  open problems, 1999.

\bibitem[SZA{\etalchar{+}}18]{shymyrbay2018meeting}
Ayan Shymyrbay, Arshyn Zhanbolatov, Assilkhan Amankhan, Adilya Bakambekova, and
  Ikechi~A Ukaegbu.
\newblock Meeting deadlines in datacenter networks: An analysis on
  deadline-aware transport layer protocols.
\newblock In {\em 2018 International Conference on Computing and Network
  Communications (CoCoNet)}, pages 152--158. IEEE, 2018.

\bibitem[VLL90]{VeltmanLL90}
B.~Veltman, B.J. Lageweg, and J.K. Lenstra.
\newblock Multiprocessor scheduling with communication delays.
\newblock {\em Parallel Computing}, 16(2):173 -- 182, 1990.

\bibitem[ZCB{\etalchar{+}}15]{zhao2015rapier}
Yangming Zhao, Kai Chen, Wei Bai, Minlan Yu, Chen Tian, Yanhui Geng, Yiming
  Zhang, Dan Li, and Sheng Wang.
\newblock Rapier: Integrating routing and scheduling for coflow-aware data
  center networks.
\newblock In {\em 2015 IEEE Conference on Computer Communications (INFOCOM)},
  pages 424--432. IEEE, 2015.

\bibitem[ZZC{\etalchar{+}}12]{zhang2012optimizing}
Jiaxing Zhang, Hucheng Zhou, Rishan Chen, Xuepeng Fan, Zhenyu Guo, Haoxiang
  Lin, Jack~Y Li, Wei Lin, Jingren Zhou, and Lidong Zhou.
\newblock Optimizing data shuffling in data-parallel computation by
  understanding user-defined functions.
\newblock In {\em Presented as part of the 9th $\{$USENIX$\}$ Symposium on
  Networked Systems Design and Implementation ($\{$NSDI$\}$ 12)}, pages
  295--308, 2012.

\end{thebibliography}

\appendix

\newpage

\section*{Appendix: The analysis of the CKR clustering}

%The FKR paper cited this claim from CKR with $|U| = 2$ (that means the probability that
%a pair of nodes is separated) and I am not seeing the full statement in CKR
In this section we reprove the statement of Theorem~\ref{thm:ProbUSeperatedByClustering}.
The claim from Theorem~\ref{thm:ProbUSeperatedByClustering}.(a) is easy to show as
\[
\textrm{diam}(V_i) = \max_{u,v \in V_i} d(u,v) \leq 2\max_{u \in V_i} \underbrace{d(u,c_i)}_{\leq \beta \Delta} \leq \Delta
\]
The tricky part is to show Theorem~\ref{thm:ProbUSeperatedByClustering}.(b).
The following definition and lemma are needed.
\begin{definition}
  Let us say that a node $w$ is a \emph{separator} for $U$, if
  \begin{enumerate*}
  \item[(A)] $\sigma(u) = w$ for at least one $u \in U$
  \item[(B)] $\sigma(u) \neq w$ for at least one $u \in U$
  \end{enumerate*}
  Moreover, if the set of separators of $U$ is non-empty, then we call the separator
  that comes first in the order $\pi$ the \emph{first separator}.
\end{definition}

Next, we show that nodes that are closer to the set $U$ are the most likely to be the first separator: 
\begin{lemma} \label{lem:CKR-ProbWsIsFirstSeparator}
  Let $w_1,\ldots,w_n$ be the nodes sorted so that $d(w_1,U) \leq \ldots \leq d(w_n,U)$.
  Then \\$\Pr[w_s\textrm{ is the first separator for }U] \leq \frac{4}{s} \cdot \frac{\textrm{diam}(U)}{\Delta}$.
\end{lemma}
\begin{proof}
  Let $u_{\min} := \textrm{argmin}\{ d(u,w_s) : u \in U\}$ and $u_{\max} := \textrm{argmax}\{ d(u,w_s) : u \in U\}$
  be the closest and furthest point from $w_s$.
  \begin{figure}
  \begin{center}
  \ifrenderfigures
    \begin{pspicture}(0,0)(6,3)
      \psellipse[linestyle=dashed](4.8,1.4)(1.45,0.8) \rput[r](3.6,2){$U$}
     \cnode*(0,0){2.5pt}{w5} \nput{-90}{w5}{$w_s$}
     \cnode*(1,0){2.5pt}{w4} \nput{-90}{w4}{$w_{s-1}$}
     \cnode*(3,0){2.5pt}{w3}
     \cnode*(4,0){2.5pt}{w2} \nput{-90}{w2}{$w_2$}
     \cnode*(5,0){2.5pt}{w1} \nput{-90}{w1}{$w_1$}
     \rput[c](2,0){$\ldots$}
     \cnode*(5,0.8){2.5pt}{u1}
     \cnode*(3.8,1){2.5pt}{u2} % \nput{-90}{w5}{$w_s$}
     \cnode*(6,1.5){2.5pt}{u3}
     \cnode*(5,2){2.5pt}{u4}
     \ncline{<->}{u2}{u3} \naput[labelsep=0pt]{$\leq \textrm{diam}(U)$}
     \nput{180}{u2}{$u_{\min}$}
     \nput{0}{u3}{$u_{\max}$}
   \end{pspicture}
   \fi
 \end{center}
 \caption{Visualization of CKR analysis}
\end{figure}
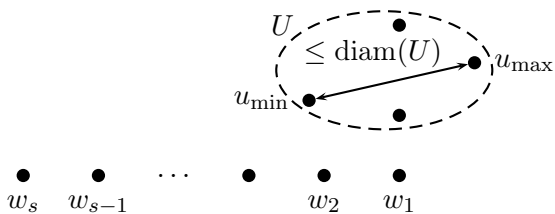
  We claim that in order for $w_s$ to be the first separator, both of the following
  conditions must hold:
\begin{enumerate*}
\item[(i)] $d(w_s,u_{\min}) \leq \beta \cdot \Delta < d(w_s,u_{\max})$
\item[(ii)] The order selects $w_s$ as the first node among $w_1,\ldots,w_s$.
\end{enumerate*}
We assume that $w_s$ is the first separator, and
suppose for the sake of contradiction that either (i) or (ii) (or both) are not satisfied.
We verify the cases:
\begin{itemize}
\item \emph{Case: $\beta \Delta < d(w_s,u_{\min})$.} Then no point will be assigned to $w_s$ and $w_s$ is not a separator at all.
\item \emph{Case: $\beta \Delta \geq d(w_s,u_{\max})$.} As $w_s$ is a separator, there are nodes $u_1,u_2 \in U$
  with $\sigma(u_1) = w_s$ and $\sigma(u_2) \neq w_s$. Then $\sigma(u_2)$ has to come earlier in the order $\pi$ as $d(w_s,u_2) \leq \beta \Delta$. Hence $w_s$ is not the first separator.
\item \emph{Case: $w_s$ is not first among $w_1,\ldots,w_s$ with respect to $\pi$.} By assumption there
  is an index $1 \leq s_2 <s$ such that $\pi(w_{s_2}) < \pi(w_s)$. As $w_s$ is a separator, there
  is a $u_1 \in U$ with $\sigma(u_1) = w_s$.
  Let $u_2 := \textrm{argmin}\{ d(u,w_{s_2}) : u \in U\}$ be the point in the set $U$ that is closest to $w_{s_2}$.
  Then $d(u_2,w_{s_2}) = d(w_{s_2},U) \leq d(w_s,U) \leq d(w_s,u_1) \leq \beta \Delta$. 
  Hence $u_2$ would be assigned to a point of order at most $\pi(w_{s'}) < \pi(w_s)$, and therefore $w_s$ is not the first separator.
\end{itemize}
Now we estimate the probability that $w_s$ is the first separator.
 The parameter $\beta$ and the permutation are chosen independently, so $(i)$ and $(ii)$ are
independent events. Clearly $\Pr[(ii)] = \frac{1}{s}$. Moreover
\[
  \Pr[(i)] = \frac{|[d(w_s,u_{\min}),d(w_s,u_{\max})] \cap [\frac{\Delta}{4},\frac{\Delta}{2}]|}{\Delta/4}
  \leq \frac{4d(u_{\min},u_{\max})}{\Delta} \leq \frac{4\textrm{diam}(U)}{\Delta},
\]
where we have used the triangle inequality and the notation  $|[a,b]| = b-a$ for the length of an interval.
\end{proof}
Now we can finish the proof of Theorem~\ref{thm:ProbUSeperatedByClustering}:
\begin{proof}[Proof of Theorem~\ref{thm:ProbUSeperatedByClustering}]
  As in Lemma~\ref{lem:CKR-ProbWsIsFirstSeparator}, let $w_1,\ldots,w_n$ be an order of nodes such that $d(w_1,U) \leq \ldots \leq d(w_n,U)$.
  Note that $L := |N(U,\frac{\Delta}{2})| \leq n$ is the maximal index with $d(w_L,U) \leq \frac{\Delta}{2}$.
  If $U$ is separated, then there has to be a first separator.
Therefore, the following holds:
  \begin{eqnarray*}
    \Pr[U\textrm{ is separated}] &\leq& \sum_{s=1}^L \Pr[w_s\textrm{ is first separator for }U] \\
    &\stackrel{\textrm{Lem~\ref{lem:CKR-ProbWsIsFirstSeparator}}}{\leq}& \underbrace{\sum_{s=1}^L \frac{1}{s} \cdot}_{\leq \ln(2L)} \frac{4\textrm{diam}(U)}{\Delta} \leq \ln(2L) \cdot \frac{4\textrm{diam}(U)}{\Delta}
  \end{eqnarray*}
  Here we use that $\Pr[w_s\textrm{ is first separator for }U] = 0$ for $s>L$ since a node $w_s$ that has a distance bigger than $\frac{\Delta}{2}$ to $U$ will never be a separator.
\end{proof}

\end{document}